\documentclass[conference]{IEEEtran}
\IEEEoverridecommandlockouts

\usepackage{cite}
\usepackage{amsthm}
\usepackage{amsmath}
\usepackage{amssymb}
\usepackage{balance}
\usepackage{graphicx}
\usepackage[hyphens]{url}
\usepackage[bookmarks=true,hidelinks,pdfstartview=FitH]{hyperref}
\usepackage{xcolor}
\hypersetup{
	colorlinks,
	linkcolor={red!50!black},
	citecolor={blue!50!black},
	urlcolor={blue!80!black}
}

\usepackage[ruled,vlined,linesnumbered]{algorithm2e}
\usepackage{enumitem}
\usepackage[caption=false,labelformat=simple]{subfig}

\SetKw{KwOr}{or}
\SetKw{KwAnd}{and}
\SetKw{KwBreak}{break}
\SetKw{KwContinue}{continue}

\usepackage{tabularx}
\graphicspath{{figs/}}

\newtheorem{theorem}{Theorem}%[section]
\newtheorem{lemma}{Lemma}%[section]

\newtheorem{definition}{Definition}%[section]
\newcommand{\e}{{\ensuremath{\mathrm{e}}}}
\newcommand{\A}{{\mathcal A}}

\newcommand{\R}{{\mathcal R}}

\newcommand{\E}{{\mathbb E}}

\newcommand{\opt}{\ensuremath{\mathrm{opt}}}

\newcommand{\fest}{\tilde{\rho}_{\mathrm{f}}}
\newcommand{\rest}{\tilde{\rho}_{\mathrm{r}}}

\newcommand{\Cov}{\operatorname{Cov}}

\newcommand{\ADG}{\textsc{ADG}\xspace}
\newcommand{\ARS}{\textsc{ARS}\xspace}
\newcommand{\ATPS}{\textsc{AddATP}\xspace}
\newcommand{\ATPA}{\textsc{HATP}\xspace}
\newcommand{\NTPA}{\textsc{HNTP}\xspace}

\newcommand{\NSG}{\textsc{NSG}\xspace}
\newcommand{\NDG}{\textsc{NDG}\xspace}

\newcommand{\indeg}{\mathrm{indeg}}

\newcommand{\etal}{{et al.}\xspace}
\newcommand{\spara}[1]{\vspace{2mm}\noindent\textbf{#1.}}
\newcommand{\ie}{{i.e.,}\xspace}

\title{Efficient Approximation Algorithms for Adaptive Target Profit Maximization\thanks{A short version of the paper will appear in the 36th IEEE International Conference on Data Engineering (ICDE '20), April 20--24, 2020, Dallas, Texas, USA.}}

\makeatletter
\newcommand{\linebreakand}{%
\end{@IEEEauthorhalign}
\hfill\mbox{}\par
\mbox{}\hfill\begin{@IEEEauthorhalign}
}
\makeatother

\author{\IEEEauthorblockN{Keke Huang}
\IEEEauthorblockA{School of Comp. Sci. and Engg.\\
Nanyang Technological University\\
khuang005@ntu.edu.sg}
\and
\IEEEauthorblockN{Jing Tang\IEEEauthorrefmark{1}\thanks{\IEEEauthorrefmark{1}Corresponding author: Jing Tang.}}
\IEEEauthorblockA{Dept. of Ind. Syst. Engg. and Mgmt.\\
National University of Singapore\\
isejtang@nus.edu.sg}
\and
\IEEEauthorblockN{Xiaokui Xiao}
\IEEEauthorblockA{School of Computing\\
National University of Singapore\\
xkxiao@nus.edu.sg}
\linebreakand
\IEEEauthorblockN{Aixin Sun}
\IEEEauthorblockA{School of Comp. Sci. and Engg.\\
Nanyang Technological University\\
axsun@ntu.edu.sg}
\and
\IEEEauthorblockN{Andrew Lim}
\IEEEauthorblockA{Dept. of Ind. Syst. Engg. and Mgmt.\\
National University of Singapore\\
isealim@nus.edu.sg}
}

\begin{document}

\maketitle

\begin{sloppy}

\begin{abstract}
Given a social network $G$, the profit maximization (PM) problem asks for a set of seed nodes to maximize the profit, i.e., revenue of influence spread less the cost of seed selection. The target profit maximization (TPM) problem, which generalizes the PM problem, aims to select a subset of seed nodes from a target user set $T$ to maximize the profit. Existing algorithms for PM mostly consider the {\it nonadaptive} setting, where all seed nodes are selected in one batch without any knowledge on how they may influence other users. In this paper, we study TPM in {\it adaptive} setting, where the seed users are selected through multiple batches, such that the selection of a batch exploits the knowledge of actual influence in the previous batches. To acquire an overall understanding, we study the adaptive TPM problem under both the {\it oracle model} and the {\it noise model}, and propose \ADG and \ATPS algorithms to address them with strong theoretical guarantees, respectively. In addition, to better handle the sampling errors under the noise model, we propose the idea of {\it hybrid error} based on which we design a novel algorithm \ATPA that boosts the efficiency of \ATPS significantly. We conduct extensive experiments on real social networks to evaluate the performance, and the experimental results strongly confirm the superiorities and effectiveness of our solutions.
\end{abstract}


\section{Introduction}\label{sec:intro}

Online social networks (OSNs) such as Facebook and Twitter have witnessed their prosperous developments in recent years. Many companies have taken OSNs as the major advertising channels to promote their products via word-of-mouth effect. Those rapid proliferations have motivated substantial research on viral marketing strategies for maximal profits. Specifically, market strategy makers seek for a proper set of influential individuals and invest in each of them (e.g., cashback rewards, coupons, or discounts) to exert their influence on advertising, aiming to maximize the expected profit. This problem is commonly studied as the {\it profit maximization (PM)} problem in the literature~\cite{Arthur_pricestrategy_2009,Lu_profitmax_2012,Zhu_infprofit_2017,Tang_Profit_2018,Tang_profitMax_2016,Tang_profitMax_2018}. PM problem asks for a set of seed users $S$ from network $G$ at the cost of $c(S)$ so as to maximize the total expected profit, \ie the expected spread of $S$ less the total investment cost $c(S)$.

However, this vanilla PM problem overlooks one fact that even though the advertisers have the full knowledge about the whole network, they are likely to only have access to a fraction of users~\cite{Badanidiyuru_adapSeeding_2016}, which is quite common in marketing applications. For example, companies advertise new products to users in their subscription mailing list, or new shop owners provide free samples to the popularities or celebrities who visit their store on site, to name a few. To circumvent this potential issue, we extend the vanilla PM problem into a more generalized version and propose the {\it target profit maximization (TPM)} problem. Specifically, given a social network $G=(V,E)$ and a target user set $T \subseteq V$ with nonnegative expected profit where each user $u \in T$ is associated with a cost $c(u)$, the TPM problem aims to select a subset $S \subseteq T$ to maximize the expected profit. The target user set $T$ can be either made up by current accessible users in the social network or selected by current seed selection algorithms. In particular, if the target set $T$ contains all users in the social network, \ie $T=V$, TPM then degenerates to PM.

Existing work on profit maximization mostly focuses on the {\it nonadaptive} setting~\cite{Arthur_pricestrategy_2009,Lu_profitmax_2012,Zhu_infprofit_2017,Tang_Profit_2018,Tang_profitMax_2016,Tang_profitMax_2018}, where all seed nodes are selected in one batch without any knowledge on how they may influence other users. As a consequence of the nonadaptiveness, the maximal possible profit might not be achieved. On the contrary, if we could observe the real-time feedback from the market and response with a smarter seed selection procedure, we could make further improvement on the final profit due to the {\it adaptivity gap}~\cite{Arthur_pricestrategy_2009, Golovin_adaptive_2011, Han_AIM_2018,Tang_ASM_2019}. Motivated by this fact, we propose an adaptive strategy on seed selection for profit maximization and formulate the problem as the {\it adaptive TPM} problem. In a nutshell, as long as there exists a profitable target seed set, adaptive TPM would try to derive a subset to maximize the profit by exploiting the adaptivity advantage.

To the best of our knowledge, we are the first to consider profit maximization in adaptive setting and aim to provide inspirational insights for future research. To acquire an overall understanding, we first study this problem in {\it oracle model} where the profit of any node set can be obtained in $O(1)$ time. We then consider it in a more practical {\it noise model} where the expected profits (or expected spreads) can only be estimated through sampling. Eventually, we design an efficient algorithm with nontrivial approximation guarantees to address the adaptive TPM problem. In summary, our major contributions are as follows.

\begin{itemize}[topsep=2mm, partopsep=0pt, itemsep=1mm, leftmargin=18pt]
	\item \textbf{Proposal of adaptive target profit maximization.} We are the first to consider profit maximization in adaptive setting and propose the adaptive target profit maximization problem. By utilizing the knowledge from the market feedback, we study adaptive strategies to select more effective seed nodes to maximize the expected profit.
	\item \textbf{Theoretical analyses on oracle and noise model.} To gain a comprehensive understanding, we study adaptive TPM in both {\it oracle model} and {\it noise model}. Specifically, in {\it oracle model} where the profit of any node set is accessible in $O(1)$ time, we propose \ADG algorithm and prove that it could achieve an $\frac{1}{3}$-approximation. In the {\it noise model} where any expected profit (or expected spread) only can be estimated, we extend \ADG into \ATPS by considering {\it additive sampling error} and derive the corresponding approximation guarantee.
	\item \textbf{Practical algorithm with optimized efficiency.} To further optimize the efficiency of \ATPS in the noise model, we propose \ATPA algorithm by adopting the idea of {\it hybrid error}. Specifically, additive error could incur prohibitive computation overhead (see Section~\ref{sec:reasonforoptimization}). To tackle this issue, we propose hybrid error which combines both relative error and additive error, fitting nicely to the scenario of nodes with diverse expected spreads. Based on this novel technique, we design \ATPA that boosts the efficiency of \ATPS significantly. Finally, its approximation guarantee and time complexity are derived.
\end{itemize}

\section{Preliminaries}\label{sec:preliminaries}

\begin{table}[!t]
	\centering
	\caption{Frequently used notations}
	\label{tbl:prelim-notations}
	\setlength{\tabcolsep}{0.5em} % for the horizontal padding
	\renewcommand{\arraystretch}{1.2}% for the vertical padding
	%	\begin{tabular}{|p{1.6cm}|p{6cm}|}\hline%{|p{1.5cm}|p{0.75cm}|p{0.75cm}|p{1.5cm}|p{2cm}|}
	\begin{tabular}{|c|m{6.5cm}|}\hline%{|p{1.5cm}|p{0.75cm}|p{0.75cm}|p{1.5cm}|p{2cm}|}
		\textbf{Notation} & \multicolumn{1}{c|}{\textbf{Description}} \\ \hline
		$G=(V,E)$ & a social network with node set $V$ and edge set $E$\\ \hline
		$n,m$ & the numbers of nodes and edges in $G$, respectively\\ \hline
		$T$ & the set of target seed nodes \\ \hline
		$k$ & the number of elements in $T$, \ie $k=|T|$ \\ \hline
		$S_i$ & the seed set selected in the $i$-th iteration \\ \hline
		$T_i$ & the seed set candidate in the $i$-th iteration \\ \hline
		$G_i$ & the $i$-th residual graph \\ \hline
		$n_i,m_i$ & the numbers of nodes and edges in $G_i$ \\ \hline		
		$I_{G_i}(u_i)$ & the number of nodes activated by node $u_i$ on $G_i$   \\ \hline
		$\Cov_{\R}(S)$ & the number of RR-sets in $\R$ that overlap $S$ \\ \hline
		$\E[I(S)]$ & the expected spread of seed set $S$ \\ \hline 
		$\phi, \Phi, \Omega$	         	   & a specific realization, a random realization, and the realization space	\\  \hline
		$\pi^\mathrm{f}, \pi^\mathrm{r}, \pi^{\opt}$			 & a front greedy policy, a rear greedy policy, and an optimal policy					\\  \hline
	\end{tabular}
\end{table}

This section presents the formal definition of the adaptive target profit maximization problem. To demonstrate the influence propagation process, we take the {\it independent cascade (IC)} model~\cite{Kempe_maxInfluence_2003} as an illustration, which is one of the extensively studied propagation models on this topic. Table~\ref{tbl:prelim-notations} summarizes the notations that are frequently used.

\subsection{Influence Propagation and Realization}\label{sec:infpro}

Let $G=(V,E)$ be a social network with a node set $V$ and a directed edge set $E$ where $n=|V|$ and $m=|E|$. For each edge $\langle u,v \rangle \in E$, $u$ is called an {\it incoming neighbor} of $v$ and $v$ is an {\it outgoing neighbor} of $u$. Each edge $e \in E$ is associated with a probability $p(e) \in (0,1]$. For simplicity, such a social network is referred to as {\it probabilistic social graph}. Given an initial node set $S \subseteq V$, the influence propagation process started by $S$ under the independent cascade (IC) model is stochastic as follows. During time interval $t_0$, all nodes in $S$ are activated and other nodes in $V\setminus S$ are inactive. At time interval $t_i$, each node activated during time interval $t_{i-1}$ has {\it one} chance to activate its inactive outgoing neighbors with the probability associated with that edge. This influence propagation process continues until no more inactive nodes can be activated. After the propagation process terminates, let $I_G(S)$ be the number of nodes activated in $G$. We call $S$ as the {\it seed set} and $I_G(S)$ as the {\it spread} of $S$ on $G$. Note that $I_G(S)$ is a random variable with respect to the propagation process. For simplicity, we use $I(S)$ to represent $I_G(S)$ by omitting the subscript $G$ unless otherwise specified.

The above process can also be interpreted by utilizing the concept of {\it realization} (aka possible world). Given a probabilistic social graph $G$, one realization, denoted as $\phi$, is a residual graph of $G$ constructed by removing each edge $e \in E$ with probability $1-p(e)$. Let $\Omega$ be the set of all possible realizations and we have $|\Omega|=2^m$ where $m=|E|$. Let $\Phi$ be a random realization randomly sampled from $\Omega$, denoted as $\Phi \sim \Omega$. Given a specific realization $\phi \in \Omega$ and any seed set $S\subseteq V$, the number of nodes that can be reached by $S$ under $\phi$ is the spread of $S$, denoted as $I_\phi(S)$. With this regard, the expected spread of $S$, denoted as $\E[I(S)]$, can be expressed as follows.
\begin{equation}\label{eqn:expectedspread}
\E[I(S)]:=\E_{\Phi\sim \Omega}[I_\Phi(S)]=\sum_{\phi \in \Omega}I_\phi(S) \cdot p(\phi),
\end{equation}
where $p(\phi)$ is the probability of $\phi$ sampled from $\Omega$.

\subsection{Adaptive Target Profit Maximization}\label{sec:profdef}

\begin{figure*}[!t]
	\centering
	\subfloat[A social graph $G_1$]{\includegraphics[width=0.2\linewidth]{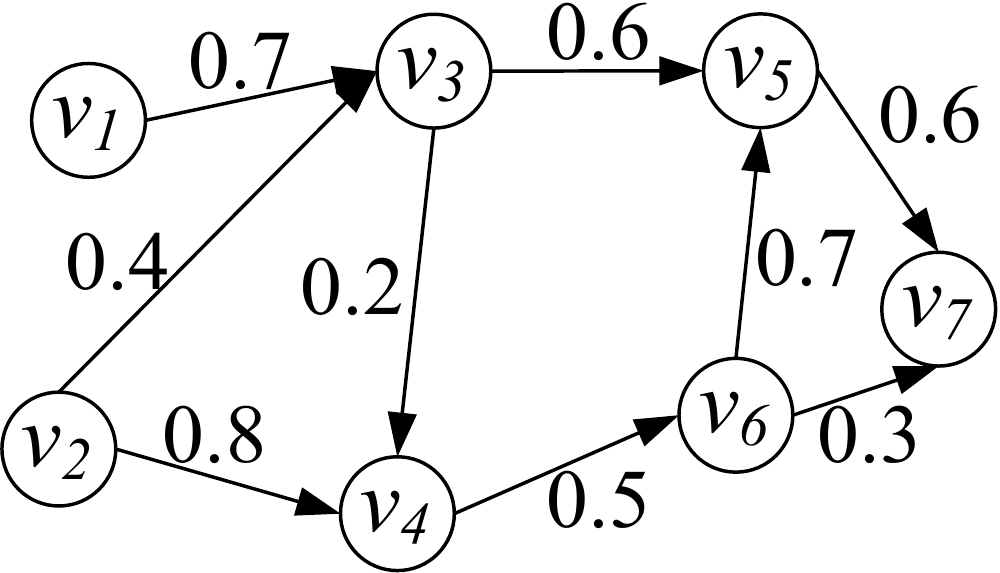}\label{subfig:a}}\hfill
	\subfloat[$v_2$ as the first seed]{\includegraphics[width=0.2\linewidth]{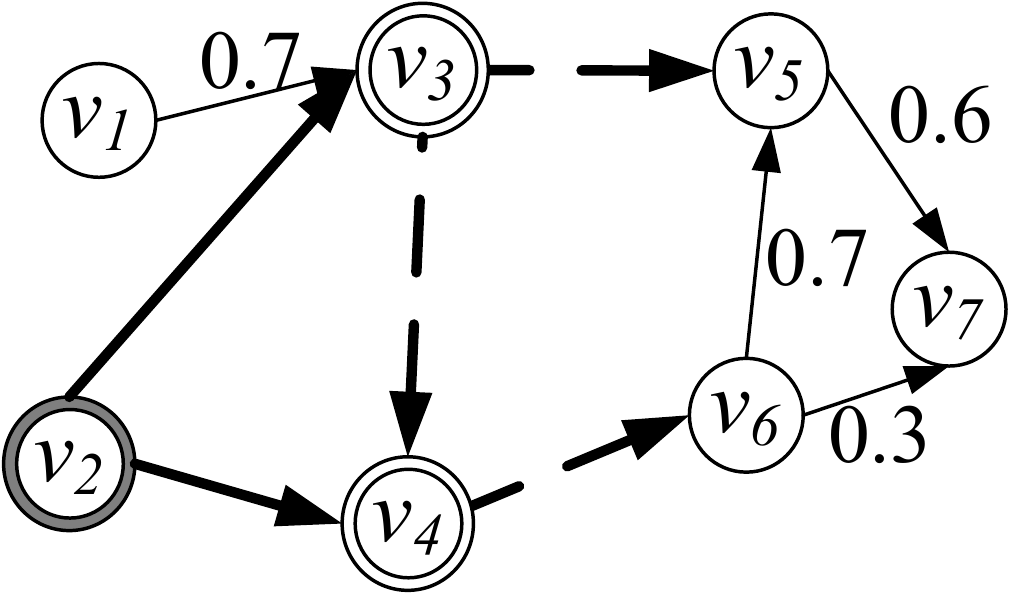}\label{subfig:b}}\hfill
	\subfloat[The residual graph $G_2$]{\includegraphics[width=0.2\linewidth]{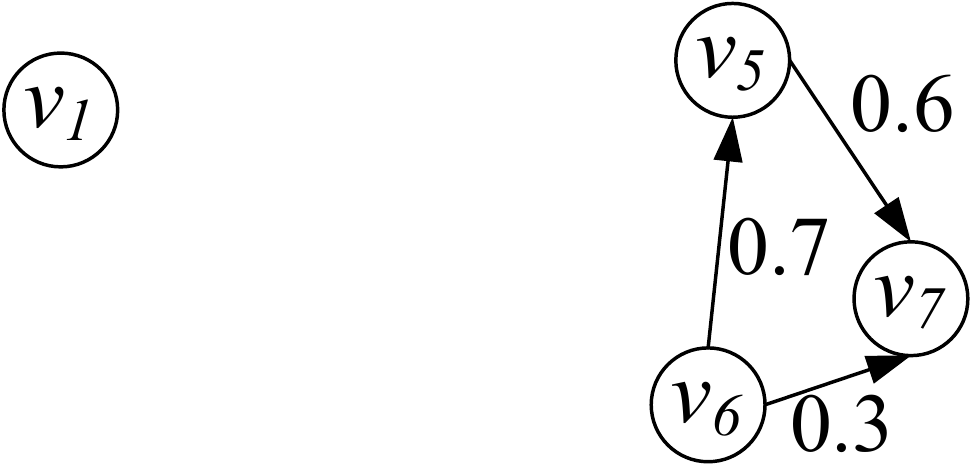}\label{subfig:c}}\hfill
	\subfloat[$v_6$ as the second seed]{\includegraphics[width=0.2\linewidth]{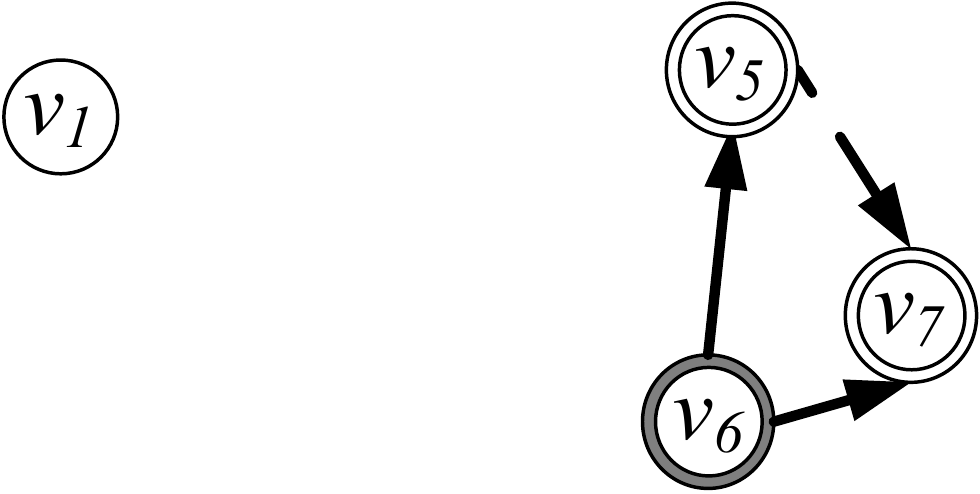}\label{subfig:d}}
	\caption{An adaptive profit maximization process.}\label{fig:APM-process}
\end{figure*}

Given a social network $G=(V,E)$, we consider a target user set $T \subseteq V$ and each user $u \in T$ is associated with a cost $c(u)$. For any set $S\subseteq T$, we denote $\rho(S)$ as the expected profit of $S$, which is defined as
\begin{equation}\label{profit}
	\rho(S):=\E[I(S)]-c(S),
\end{equation}
where $\E[I(S)]$ is the expected spread of $S$ and $c(S)=\sum_{u\in S}c(u)$. As $\rho(\cdot)$ is a positive linear combination of a monotone submodular function (\ie~$\E[I(\cdot)]$) and a modular function (\ie~$c[\cdot]$), we know that $\rho(\cdot)$ is still submodular, which, however, may not be monotone. We assume that the expected profit of the target set $T$ is nonnegative, \ie~$\rho(T)\geq 0$. Then, the target profit maximization (TPM) problem aims to select a subset $S \subseteq T$ to maximize expected profit. 

Conventionally, seed nodes are nonadaptively selected without any knowledge of realization happened in the actual influence propagation process. Contrarily, the adaptive strategy selects each node based on the real-time feedback from the realization observed, which is proved to be more effective~\cite{Arthur_pricestrategy_2009,Golovin_adaptive_2011,Han_AIM_2018}. Specifically, the adaptive strategy selects node one by one in an adaptive manner as follows. Let $G_1=G$ and $T_0=T$ at the beginning. It first {\it selects} a node $u_1$ from the target set $T_0$ based on $G_1$, and then {\it observes} how node $u_1$ activates other nodes. As follows, it removes those activated nodes from $G_1$ and $T_0$, resulting in the residual graph $G_2$ and set $T_1$ respectively. Subsequently, it then selects the next node $u_2$ from $T_2$ based on the residual graph $G_2$. This process will be repeated until a solution $S$ is produced.

For example, \figurename~\ref{subfig:a} gives a probabilistic social graph $G_1$ under the IC model. Suppose that the target seed set is $T=\{v_1, v_2, v_6\}$ and each node in $T$ has a cost of $1.5$. One can verify that the optimal nonadaptive solution is $T$ with an expected profit of $\rho_{G_1}(T)=\E[I_{G_1}(T)]-c(T)=6.16-4.5=1.66$. As with the adaptive seeding strategy, it would select node $v_2$ as the first seed node, as shown in \figurename~\ref{subfig:b}. Specifically, nodes in gray double-cycle are seeds and nodes in blank double-cycle are activated by the seeds. Bold full-line arrow (resp.\ bold dashed-line arrow) indicates a successful (resp.\ failed) step of influence. Observe that $v_3$ and $v_4$ are activated by $v_2$ and thus $G_1$ is updated to $G_2$ by removing $\{v_2$, $v_3$, $v_4\}$, as shown in \figurename~\ref{subfig:c}. Then, node $v_6$ is selected as the second seed and $v_5$ and $v_7$ will be activated, as shown in \figurename~\ref{subfig:d}. Eventually, $\{v_2, v_6\}$ are selected with a total profit of $6-3=3$. However, under this realization, the solution $\{v_1, v_2, v_6\}$ selected by the nonadaptive algorithm produces a spread of $7$, which results in the final profit of $7-4.5=2.5$. Thus, the adaptive strategy generates $20\%$ more profits.

In this paper, we aim to study algorithms that obtain seed selection strategies (also known as \textit{policies}) for adaptive target profit maximization. For convenience, we abuse notation and use the terms ``policy'' and ``algorithm'' interchangeably in the paper. Let $S_{\phi}(\pi)$ be the seed set selected by the policy $\pi$ under realization $\phi$. Then, the profit of policy $\pi$ under realization $\phi$ is $\rho_{\phi}(S_{\phi}(\pi))=I_{\phi}(S_{\phi}(\pi))-c(S_{\phi}(\pi))$. Now, we give the expected profit of policy over all possible realizations.
\begin{definition}[Expected Profit of Policy]\label{def:polprofit}
The expected profit $\Lambda(\pi)$ of policy $\pi$ is defined as 
\begin{equation}
	\Lambda(\pi):=\E_{\Phi \sim \Omega}[\rho_{\Phi}(S_{\phi}(\pi))]=\sum_{\phi \in \Omega}\rho_{\phi}(S_{\phi}(\pi)) \cdot p(\phi),
\end{equation}
where $p(\phi)$ is the probability of $\phi$ sampled from $\Omega$.
\end{definition}
Based on the definition of expected profit of policy, adaptive target profit maximization is defined as follows.
\begin{definition}[Adaptive TPM]\label{def:adaptTarProfMax}
Given a probabilistic graph $G=(V,E)$ and a target user set $T\subseteq V$ with nonnegative expected profit where each node $u \in T$ is associated with a cost $c(u)$ for seed selection, the adaptive target profit maximization problem asks for a policy $\pi^\ast$ that maximizes the expected profit over all possible realizations, \ie 
\begin{equation}
	\pi^\ast:={\arg \max}_{\pi} \Lambda(\pi).
\end{equation}
\end{definition}
\section{Adaptive Target Profit Maximization}\label{sec:sec3}

In this section, we first introduce the {\it double greedy} method which is widely adopted to address the profit maximization problem in the literature. Then we extend double greedy to its adaptive version based on which we develop \ADG algorithm and \ATPS algorithm to address the adaptive target profit maximization problem under the {\it oracle model} and {\it noise model} respectively. As what follows, we then present the analyses of the approximation guarantees of \ADG and \ATPS.

\subsection{Double Greedy}\label{sec:doublegreedy}
Profit maximization is an application of {\it unconstrained submodular maximization (USM)} for which Buchbinder~\etal~\cite{Buchbinder_DoubleGreedy_2012} propose a {\it double greedy} algorithm as shown in Algorithm~\ref{alg:DG}. The double greedy algorithm is widely used to address the profit maximization problem in the literature~\cite{Tang_Profit_2018,Amo_CouponAdvertise_2018,Zhiyi_pricecompete_2018}.

Algorithm~\ref{alg:DG} presents the pseudocode of double greedy. For any submodular function $f(\cdot)$, double greedy works as follows. It maintains two sets $S$ and $T$ initialized with empty set and ground set $V$, respectively. It then checks each node $u\in V$ one by one (in arbitrary sequence) to decide on selecting or abandoning $u$. To this end, it calculates the marginal gain of selecting it conditioned on current seed set $S$, denoted as $z^+$, and the marginal gain of abandoning it conditioned on current candidate set $T$, denoted as $z^-$. If $z^+ \ge z^-$, node $u$ will be selected and added to $S$. Otherwise, $u$ is removed from $T$. When all nodes in $V$ have been checked, the resultant $S$ and $T$ are equal and either of them is returned. The rationale behind double greedy is quite straightforward, \ie the node is selected only if it could bring larger marginal gain by keeping it than that by removing it.
\begin{algorithm}[!t]
	\caption{{Double Greedy}~\cite{Buchbinder_DoubleGreedy_2012} \label{alg:DG}}	
	Initialize $S\leftarrow\emptyset$, $T\leftarrow V$\; 
	\ForEach{node $u\in V$}
	{
		$z^+ \leftarrow f(u \mid S)$\;\label{alg0:frontprof}
		$z^- \leftarrow -f(u \mid T \setminus \{u\})$\;\label{alg0:rearprof}
		\lIf{$z^+ \ge z^-$}{$S\leftarrow S\cup \{u\}$}
		\lElse{$T\leftarrow T\setminus\{u_i\}$}
	}
	\Return $S\ (=T)$\;
\end{algorithm}

\subsection{Adaptive Double Greedy under the Oracle Model}\label{sec:adaptoracle}

\subsubsection{\bf{Description of \ADG}}\label{sec:adgdescrip}
Under the oracle model, we assume that the expected spread of any node set is accessible in $O(1)$ time. Based on this assumption, we design {\it adaptive double greedy (\ADG)}. The decision on seed selection of \ADG relies on the conditional expected marginal profit.

\begin{definition}[Conditional Expected Marginal Profit]\label{def:ConExpMarProf}
	Given a node $u$ and a node set $S$, the conditional expected marginal profit of $u$ conditioned on $S$ on graph $G$ is defined as
	\begin{equation}
	\Delta_{G}(u\mid S):=\rho_{G}(S\cup\{u\})-\rho_{G}(S).
	\end{equation}
	Then, we have $\Delta_{G}(u\mid S)=\E[I_{G}(u\mid S)] -c(u)$ if $u\not\in S$, and $\Delta_{G}(u\mid S)=0$ otherwise.
\end{definition}

Algorithm~\ref{alg:adaptoracle} presents the pseudocode of \ADG. First, two sets $S_0$ and $T_0$ are initialized with empty set and the target set $T$ respectively (Line~\ref{alg1:initial}), where $T$ contains $k$ nodes. In the $i$-th iteration, if the examining node $u_i$ is activated already, it just checks the next node (Lines~\ref{alg1:check}--\ref{alg1:continue}). Otherwise, it calculates the conditional expected marginal profit $\Delta_{G_i}(u_i\mid S_{i-1})$ (Line~\ref{alg1:frontprof}), denoted as {\it front profit $\rho_\mathrm{f}$}. Meanwhile, it also calculates $-\Delta_{G_i}(u_i\mid T_{i-1} \setminus \{u_i\})$ (Line~\ref{alg1:rearprof}), denoted as {\it rear profit $\rho_\mathrm{r}$}. If $\rho_\mathrm{f} \ge \rho_\mathrm{r}$, it would (i) insert $u_i$ into $S_i$ as one seed node, (ii) observe the set of nodes $A(u_i)$ activated by $u_i$, and (iii) update $G_i$ into $G_{i+1}$ by removing all nodes in $A(u_i)$. Otherwise, $u_i$ is removed from $T_{i-1}$. This process is repeated until all $k$ nodes in $T$ are examined and then the resultant $S_k$ is returned. Note that when \ADG terminates, $S_k$ and $T_k$ contain exactly the same nodes, \ie~$S_k=T_k$.
\begin{algorithm}[!t]
	\KwIn{Social graph $G$, a target seed set $T$ with $k$ nodes}
	\KwOut{Selected seed node set $S_k$}
	\caption{\ADG}
	\label{alg:adaptoracle}	
	Initialize $S_0\leftarrow\emptyset$, $T_0\leftarrow T$, $G_1\leftarrow G$\; \label{alg1:initial}
	\For{$i\leftarrow 1$ \KwTo $k$}
	{
		\If{$u_i$ is activated\label{alg1:check}}
		{
			$T_i\leftarrow T_{i-1}\setminus\{u_i\}$, $S_i\leftarrow S_{i-1}$, $G_{i+1}\leftarrow G_i$\;
			\KwContinue\;\label{alg1:continue}
		}
		$\rho_\mathrm{f}\leftarrow\Delta_{G_i}(u_i \mid S_{i-1})$\;\label{alg1:frontprof}
		$\rho_\mathrm{r}\leftarrow-\Delta_{G_i}(u_i \mid T_{i-1} \setminus \{u_i\})$\;\label{alg1:rearprof}
		\If{$\rho_\mathrm{f} \ge \rho_\mathrm{r}$}
		{
			$S_i\leftarrow S_{i-1}\cup \{u_i\}$, $T_i \leftarrow T_{i-1}$\;
			Observe the node set $A(u_i)$ activated by $u_i$\;
			Update $G_i$ into $G_{i+1}$ by removing $A(u_i)$\;
		}
		\lElse
		{$T_i\leftarrow T_{i-1}\setminus\{u_i\}$, $S_i\leftarrow S_{i-1}$, $G_{i+1}\leftarrow G_i$}
	}
	\Return $S_k$\;
\end{algorithm}

As shown, the decision on seed selection in \ADG is fully determined by the value of $\rho_\mathrm{f}$ and $\rho_\mathrm{r}$. In a nutshell, for each node $u\in T$, $u$ will be selected only if selecting $u$ could bring more marginal profit than abandoning it.

\subsubsection{\bf{Approximation Guarantee}}\label{sec:alg1appro} 
In this section, we would explore the approximation guarantee that \ADG could achieve. First, we show the relation between the front profit $\rho_\mathrm{f}$ and the rear profit $\rho_\mathrm{r}$ in Algorithm~\ref{alg:adaptoracle}.
\begin{lemma}\label{lem:frsum}
	If $u_i$ is inactive, we have $\rho_\mathrm{f}+\rho_\mathrm{r} \ge 0$.
\end{lemma}
\begin{proof}[Proof of Lemma~\ref{lem:frsum}]
	Given a residual graph $G_i$ and the corresponding node $u_i$, $\rho_\mathrm{f}=\E[I_{G_i}(u_i\mid S_{i-1})]-c(u_i)$ and $\rho_\mathrm{r}=c(u_i)-\E[I_{G_i}(u_i\mid T_{i-1} \setminus \{u_i\})]$. Then we have $\rho_\mathrm{f}+\rho_\mathrm{r}=\E[I_{G_i}(u_i\mid S_{i-1})]-\E[I_{G_i}(u_i\mid T_{i-1}\setminus \{u_i\})]$. Since $S_{i-1} \subseteq T_{i-1}\setminus \{u_i\}$ and $u_i \in T \setminus (T_{i-1}\setminus \{u_i\})$, based on the property of submodularity, we have $\rho_\mathrm{f}+\rho_\mathrm{r} \ge 0$, which completes the proof.
\end{proof}

To facilitate the analysis that follows, we define the notions of {\it policy truncation}, {\it policy concatenation}, and {\it policy intersection}, which are conceptual operations performed by a policy. Note that these policy operations are used for our theoretical analysis only, and they do not affect the actual implementation of our algorithms.

\begin{definition}[Policy Truncation~\cite{Golovin_adaptive_2011}]\label{def:potrun}
	For any adaptive seeding policy $\pi$, the policy truncation $\pi_{[i]}$ denotes an adaptive policy that performs exactly the same as $\pi$, except that $\pi_{[i]}$ only evaluates the first $i$ ($i\leq n$) nodes in a given node sequence.
\end{definition}

\begin{definition}[Policy Concatenation~\cite{Golovin_adaptive_2011}]\label{def:pocan}
	For any two adaptive seeding policies $\pi$ and $\pi^\prime$, the policy concatenation $\pi\oplus\pi^\prime$ denotes an adaptive policy that first executes the policy $\pi$, and then executes $\pi^\prime$ as if from a fresh start without any knowledge about $\pi$.
\end{definition}

\begin{definition}[Policy Intersection]\label{def:point}
	For any two adaptive seeding policies $\pi$ and $\pi^\prime$, the policy intersection $\pi\otimes\pi^\prime$ denotes an adaptive policy that executes the intersection part of policy $\pi$ and $\pi^\prime$.
\end{definition}

Note that under any given realization $\phi$, policy concatenation shows that $S_\phi(\pi\oplus\pi^\prime)=S_\phi(\pi)\cup S_\phi(\pi^\prime)$ and policy intersection shows that $S_\phi(\pi\otimes\pi^\prime)=S_\phi(\pi)\cap S_\phi(\pi^\prime)$. Let $\pi^{\opt}$ be the optimal policy for adaptive TPM problem. Let $\pi^\mathrm{f}$ (resp.\ $\pi^\mathrm{r}$) be the {\it front greedy} (resp.\ {\it rear greedy}) policy of \ADG that executes as same as \ADG and selects a set $S_{i}$ (resp. $T_{i}$) of nodes after the $i$-th iteration, \ie~$S(\pi^\mathrm{f}_{[i]})=S_i$ (resp.\ $S(\pi^\mathrm{r}_{[i]})=T_i$). Define policy $\pi^\circ=(\pi^{\opt}\oplus\pi^\mathrm{f}) \otimes \pi^\mathrm{r}$ and its truncation as $\pi^\circ_{[i]}=(\pi^{\opt}\oplus\pi^\mathrm{f}_{[i]}) \otimes \pi^\mathrm{r}_{[i]}$. Let $S^\circ$ be the solution obtained by $\pi^\circ$, and $S^\circ_{i}$ be the temporal seed set obtained by $\pi^\circ_{[i]}$, \ie~$S^\circ_{i}=S(\pi^\circ_{[i]})=(S^\circ\cup S_{i})\cap T_{i}$. Then, we have $\Lambda(\pi^\mathrm{f}_{[0]})=0$, $\Lambda(\pi^\mathrm{r}_{[0]})=\rho(T)\geq 0$, $\Lambda(\pi^\circ_{[0]})=\Lambda(\pi^\opt)$, and $\Lambda(\pi^\mathrm{f}_{[k]})=\Lambda(\pi^\mathrm{r}_{[k]})=\Lambda(\pi^\circ_{[k]})$, where $k=|T|$. First of all, we bound the profit achievement of policies $\pi^\mathrm{f}$ and $\pi^\mathrm{r}$ on residual graph $G_i$ in the following lemma.
\begin{lemma}\label{lem:case1+2}
For the $i$-th iteration of \ADG, we have 
\begin{equation}\label{eqn:corestep}
\begin{split}
&\rho_{G_i}(S^\circ_{i-1}) - \rho_{G_i}(S^\circ_{i})\\
&\leq \rho_{G_i}(S_{i})-\rho_{G_i}(S_{i-1}) + \rho_{G_i}(T_{i}) - \rho_{G_i}(T_{i-1}). 
\end{split}
\end{equation} 
\end{lemma}

The formal proofs of most theoretical results are given in Appendix~\ref{append:proofs}. Based on Lemma~\ref{lem:case1+2}, we further derive the following lemma for the expected profit of policies $\pi^\mathrm{f}$, $\pi^\mathrm{r}$, and policy $\pi^\circ$.
\begin{lemma}\label{lem:policyprofit}
For the $i$-th iteration of \ADG, we have 
\begin{equation*}%\label{eqn:corestep}
\Lambda(\pi^\circ_{[i-1]}) - \Lambda(\pi^\circ_{[i]})\le \Lambda(\pi^\mathrm{f}_{[i]}) - \Lambda(\pi^\mathrm{f}_{[i-1]}) + \Lambda(\pi^\mathrm{r}_{[i]}) - \Lambda(\pi^\mathrm{r}_{[i-1]}). 
\end{equation*} 
%\begin{equation}\label{eqn:corestep}
%\Lambda(\pi^\circ_{i-1}) - \Lambda(\pi^\circ_{i})\le \Lambda(\pi^\mathrm{f}_{i}) - \Lambda(\pi^\mathrm{f}_{i-1}) + \Lambda(\pi^\mathrm{r}_{i}) - \Lambda(\pi^\mathrm{r}_{i-1}). 
%\end{equation}
\end{lemma}
Lemma~\ref{lem:policyprofit} establishes the relation on expected profit of policy between \ADG and the optimal policy $\pi^\opt$ for each iteration. Based on Lemma~\ref{lem:policyprofit}, the approximation guarantee of \ADG is derived as follows.
\begin{theorem}\label{thm:adaptgreedy}
\ADG achieves the approximation ratio of $1/3$.
\end{theorem}
\begin{proof}[Proof of Theorem~\ref{thm:adaptgreedy}]
From Lemma~\ref{lem:policyprofit}, we have
\begin{align*}
	&\Lambda(\pi^\circ_{[0]})-\Lambda(\pi^\circ_{[k]})
	=\sum\nolimits_{i=0}^{k}\big(\Lambda(\pi^\circ_{[i-1]}) - \Lambda(\pi^\circ_{[i]})\big)\\
	&\leq \sum\nolimits_{i=0}^{k}\big(\Lambda(\pi^\mathrm{f}_{[i]}) - \Lambda(\pi^\mathrm{f}_{[i-1]}) + \Lambda(\pi^\mathrm{r}_{[i]}) - \Lambda(\pi^\mathrm{r}_{[i-1]})\big)\\
	&=\Lambda(\pi^\mathrm{f}_{[k]})-\Lambda(\pi^\mathrm{f}_{[0]})+\Lambda(\pi^\mathrm{r}_{[k]})-\Lambda(\pi^\mathrm{r}_{[0]}).
\end{align*}
Recalling that $\Lambda(\pi^\mathrm{f}_{[0]})=0$, $\Lambda(\pi^\mathrm{r}_{[0]})=\rho(T)\geq 0$, $\Lambda(\pi^\circ_{[0]})=\Lambda(\pi^\opt)$, and $\Lambda(\pi^\mathrm{f}_{[k]})=\Lambda(\pi^\mathrm{r}_{[k]})=\Lambda(\pi^\circ_{[k]})$, we obtain that $\Lambda(\pi^\mathrm{f}_{[k]})=\Lambda(\pi^\mathrm{r}_{[k]})\geq \Lambda(\pi^\opt)/3$.
\end{proof}

\spara{Remark} At the first glance, the equations and transformations in this paper might look similar to those in the existing influence maximization papers, which is because we adopt the commonly used notations and definitions in the literature. However, we note that influence maximization is based on monotone submodular optimization with a cardinality constraint, whereas profit maximization is based on unconstrained (i.e., nonmonotone) submodular optimization. Due to the fundamental difference between the two problems, our algorithms and mathematical analysis (including the equations and transformations) actually {\it differ} considerably from those in the existing influence maximization papers. Specifically, a simple greedy algorithm is used to address influence maximization, while we devise an adaptive \textit{double} greedy algorithm tailored for adaptive target profit maximization.

\subsection{Adaptive Double Greedy under the Noise Model}\label{sec:Noisemodel}

As well-known, computing the exact expected spread of any node set is \#P-hard~\cite{Chen_LDAG_2010}. In this section, we try to estimate the expected spread of any seed set by taking the sampling error into account. In general, there are two forms of sampling error, \ie~{\it relative error} and {\it additive error}. Considering that we need to estimate the expected marginal spreads of $2k$ node sets during the whole process, it can be quite intricate utilizing relative error. To explain, for those with small expected marginal spreads, only a trivial amount of estimation error would be allowed, which is rather difficult to achieve by existing methods for spread estimation. Motivated by this, we adopt the additive error instead and propose the \ATPS \footnote{Algorithms with \underline{add}itive error for \underline{a}daptive \underline{t}argeted \underline{p}rofit maximization} algorithm. Algorithm~\ref{alg:AdaptAddiError} presents the details of \ATPS. 

\subsubsection{\bf {Description of \ATPS}}\label{sec:atpsdescrip}

\begin{algorithm}[!t]
	\setlength{\hsize}{0.963\linewidth}
	\KwIn{Social graph $G$, a target seed set $T$ with $k$ nodes, the initial error $\zeta_0$}
	\KwOut{Selected seed node set $S_k$}
	\caption{{\ATPS} \label{alg:AdaptAddiError}}	
	Initialize $S_0\leftarrow \emptyset$, $T_0\leftarrow T$\;\label{alg2:initial}	
	\For{$i\leftarrow 1$ \KwTo $k$}
	{
		\If{$u_i$ is activated}
		{$T_i\leftarrow T_{i-1}\setminus\{u_i\}$, $S_i\leftarrow S_{i-1}$, $G_{i+1}\leftarrow G_i$\; 
			\KwContinue\;}
		$\zeta_i\leftarrow \zeta_0$, $\delta_i\leftarrow 1/(kn)$\label{alg2:initial2}\tcp*[r]{$\zeta_0\geq 1/n$}
		\While{true}
		{
			$\theta\leftarrow \frac{1}{2\zeta_i^2}\ln\frac{8}{\delta_i}$\;
			Generate $\theta$ RR sets as $\R_1$ and $\R_2$, respectively\; 
			$\fest\leftarrow\Cov_{\R_1}(u_i\mid S_{i-1})\cdot \frac{n_i}{\theta}-c(u_i)$\;\label{alg2:fest}
			$\rest\leftarrow-\Cov_{\R_2}(u_i\mid T_{i-1} \setminus \{u_i\}) \cdot \frac{n_i}{\theta}+c(u_i)$\;\label{alg2:rest}
			\If{$C_1$ \KwOr $C_2$\label{alg2:judgement}}
			{
				\If{$\fest\geq \rest$\label{alg2:select}}
				{
					$S_i\leftarrow S_{i-1}\cup \{u_i\}$, $T_i \leftarrow T_{i-1}$\;
					Observe the node set $A(u_i)$ activated by $u_i$\;
					Update $G_i$ into $G_{i+1}$ by removing $A(u_i)$\;	
				}
				\lElse{$T_i\leftarrow T_{i-1}\setminus\{u_i\}$, $S_i\leftarrow S_{i-1}$, $G_{i+1}\leftarrow G_i$}\label{alg2:abandon}
				\KwBreak\;
			}		
			$\zeta_i\leftarrow\zeta_i/\sqrt{2}$, $\delta_i\leftarrow\delta_i/2$\;\label{alg2:varaiblehalf}
		}
	}
	\Return $S_k$\;
\end{algorithm}

The design principle of \ATPS follows that of \ADG except that $\rho_\mathrm{f}$ and $\rho_\mathrm{r}$ in \ATPS are not accessible but to be estimated. The estimation reliability is guaranteed by {\it Hoeffding Inequality}~\cite{Hoeffding_MINTSS_1963}. Specifically, \ATPS first initializes $S_0$ with $\emptyset$ and $T_0$ with target set $T$ (Line~\ref{alg2:initial}). In each iteration on residual graph $G_i$, it first checks if the current node $u_i$ is activated. If $u_i$ is activated, it skips current node and starts the next iteration immediately. Otherwise, \ATPS initializes the additive error parameter $\zeta_i=\zeta_0$ and the probability parameter $\delta_i=1/(kn)$ (Line~\ref{alg2:initial2}) for the following evaluations. Notice that a proper initialization value $\zeta_0$ could speed up the performance empirically because some nodes can be decided within relatively small number of samples. To estimate the expected marginal spreads, \ATPS adopts the commonly used {\it reverse influence sampling (RIS)~\cite{Borgs_RIS_2014}} technique, referred to as {\it RR sets}. Specifically, \ATPS generates $\theta$ RR sets into $\R_1$ and $\R_2$ to calculate $\fest$ (Line~\ref{alg2:fest}) and $\rest$ (Line~\ref{alg2:rest}) respectively, where $\theta$ is determined by $\zeta_i$ and $\delta_i$, $\fest$ and $\rest$ are the estimations of $\E[I_{G_i}(u_i\mid S_{i-1})]-c(u_i)$ and $c(u_i)-\E[I_{G_i}(u_i\mid T_{i-1} \setminus \{u_i\})]$. By Hoeffding Inequality~\cite{Hoeffding_MINTSS_1963}, we have $\rho_\mathrm{f}\in[\fest-n_i\zeta_i, \fest+n_i\zeta_i]$ and $\rho_\mathrm{r}\in[\rest-n_i\zeta_i, \rest+n_i\zeta_i]$ with high probability. Then, the stopping conditions $C_1$ and $C_2$ are checked (Line~\ref{alg2:judgement}), where $C_1$ and $C_2$ are defined as
\begin{align*}
&C_1\colon (|\fest-\rest|\geq 2n_i\zeta_i)\vee (\fest\leq -n_i\zeta_i)\vee (\rest\leq -n_i\zeta_i),\\
&C_2\colon n_i\zeta_i \leq 1.
\end{align*}
If the stopping condition $C_1$ is met, it indicates that the current estimations $\fest$ and $\rest$ are accurate enough to help make the right decision with high probability. If $C_2$ is observed instead, it indicates that (i) the expected spread of current node is too close to the judgement bar to distinguish, and (ii) the profit loss is insignificant if wrong decision is made. The purpose of this condition is to avoid the unnecessarily prohibitive sampling overhead for sufficient accurate of spread estimations. When one of the stopping conditions is met, the decision is made accordingly (Lines~\ref{alg2:select}--\ref{alg2:abandon}). Otherwise, $\zeta_i$ (resp.\ $\delta$) is divided by $\sqrt{2}$ (resp.\ $2$) (Line~\ref{alg2:varaiblehalf}) to generate more samples for more accurate estimations. This process terminates when all $k$ nodes in $T$ are checked through.

\subsubsection{\bf{Theoretical Analysis of \ATPS}}\label{sec:adaptaddtive}
In what follows, we analyze the approximation guarantee and time complexity of \ATPS.

\spara{Approximation Guarantee} First, we present the {\it Hoeffding Inequality}~\cite{Hoeffding_MINTSS_1963} based on which the estimation is reliable with high probability.
\begin{lemma}[Hoeffding Inequality~\cite{Hoeffding_MINTSS_1963}]\label{lem:hoeffding}
Let $X_i$ be an independent bounded random variable such that for each $1 \le i\le \theta$, $X_i \in [a_i,b_i]$. Let $X=\frac{1}{\theta}\sum_{i=1}^{\theta}X_i$. Given $\zeta \in (0,1)$, then
\begin{equation}\label{eqn:hoeffding}
\Pr[\lvert X-\E[X]\rvert \ge \zeta]\le 2\e^{-\frac{2\theta^2\zeta^2}{\sum_{i=1}^{\theta}(b_i-a_i)^2}}.
\end{equation}
\end{lemma}

Similarly, to obtain the approximation ratio of \ATPS, we need to establish the relation between $\Lambda(\pi^\circ_{[i-1]}) - \Lambda(\pi^\circ_{[i]})$ and $\Lambda(\pi^\mathrm{f}_{[i]})-\Lambda(\pi^\mathrm{f}_{[i-1]})+\Lambda(\pi^\mathrm{r}_{[i]}) - \Lambda(\pi^\mathrm{r}_{[i-1]})$ as the one in Lemma~\ref{lem:policyprofit}. We note that \ADG is a {\it deterministic} algorithm while \ATPS is a {\it randomized} algorithm. To explain, \ADG has access to the expected spread of any node set under the oracle model. Contrarily, under the noise model, \ATPS adopts the randomized {\it reverse influence sampling (RIS)~\cite{Borgs_RIS_2014}} technique to estimate the expected spread, which makes $\pi^\mathrm{f}$, $\pi^\mathrm{r}$, and $\pi^\circ$ all random polices. Consequently, we need to tackle the internal randomness of \ATPS. 
\begin{lemma}\label{lem:rho-noise}
For the $i$-th iteration of \ATPS, we have 
\begin{equation}\label{eqn:corestep-noise}
\begin{split}
&\E_{\mathcal{A}}\big[\rho_{G_i}(S^\circ_{i-1}) - \rho_{G_i}(S^\circ_{i})\big]-(2+2/k)\\
&\leq \E_{\mathcal{A}}\big[\rho_{G_i}(S_{i})-\rho_{G_i}(S_{i-1}) + \rho_{G_i}(T_{i}) - \rho_{G_i}(T_{i-1})\big]. 
\end{split}
\end{equation} 
\end{lemma}

Thus, for the expected profit of policy $\pi^\mathrm{f}$, $\pi^\mathrm{r}$, and $\pi^\circ$ in the $i$-th iteration, we have following lemma.

\begin{lemma}\label{lem:policyprofitadd}
For the $i$-th iteration of \ATPS, we have
\begin{equation}\label{eqn:policyprofit3}
\begin{split}
&\E_{\A}\big[\Lambda(\pi^\circ_{[i-1]})- \Lambda(\pi^\circ_{[i]})\big]-(2+2/k) \\
&\le \E_{\A}\big[\Lambda(\pi^\mathrm{f}_{[i]})- \Lambda(\pi^\mathrm{f}_{[i-1]})+\Lambda(\pi^\mathrm{r}_{[i]})-\Lambda(\pi^\mathrm{r}_{[i-1]})\big],
\end{split}
\end{equation}
where the expectation $\E_{\A}[\cdot]$ is over the internal randomness of \ATPS and $k=|T|$.
\end{lemma}
Note that there are additional terms $2/k$ and $2$ in~\eqref{eqn:policyprofit3}. Specifically, $2/k$ is the compensation factor on the profit loss when \ATPS makes the wrong choice on $u_i$ due to the failed spread estimation. $2$ is the upper bound of the profit loss incurred when \ATPS terminates due to the stopping condition $C_2$. 

Based on Lemma~\ref{lem:policyprofitadd}, we could derive the approximation ratio of \ATPS as follows.
\begin{theorem}\label{thm:adaptwitherror}
\ATPS could achieve the expected profit at least $\frac{\Lambda(\pi^{\opt})- (2k+2)}{3}$, where $\Lambda(\pi^{\opt})$ is the expected profit of optimal policy $\pi^{\opt}$ and $k$ is the number of nodes in $T$.
\end{theorem}
\begin{proof}[Proof of Theorem~\ref{thm:adaptwitherror}]
According to Lemma~\ref{lem:policyprofitadd}, accumulating both sides of \eqref{eqn:policyprofit3} from $i=1$ to $k$ gives
\begin{align*}
&\E_{\A}\big[\Lambda(\pi^\circ_{[0]})-\Lambda(\pi^\circ_{[k]})\big]-\sum\nolimits_{i=1}^{k}(2+2/k)\\
&\leq\E_{\A}\big[\Lambda(\pi^\mathrm{f}_{[k]})-\Lambda(\pi^\mathrm{f}_{[0]})+\Lambda(\pi^\mathrm{r}_{[k]})-\Lambda(\pi^\mathrm{r}_{[0]})\big].
\end{align*}
Rearranging it completes the proof.
\end{proof}

\spara{Discussion} On particular social graphs, our algorithm could achieve an expected ratio of $(1-\varepsilon)/3$. The main idea is to dynamically set the threshold of $n_i\zeta_i$ for stopping condition $C_2$ instead of fixing to $1$. Specifically, in each iteration, we have (at most) a profit loss of $2/k$ due to a failure estimation of marginal expected spread and an extra profit loss of $2$ if stopping condition $C_2$, \ie~$n_i\zeta_i \leq1$, occurs (see the proof of Lemma~\ref{lem:rho-noise}). Thus, we have (at most) a total profit loss of $2+2k$ for all $k$ iterations. However, we note that it is unlikely that all $k$ iterations meet $C_2$. Thus, we could try to bound the actual profit loss within $\varepsilon\Lambda(\pi^{\opt})$ by adjusting the error threshold dynamically as follows. For the $i$-th iteration, let $\eta_i$ be the settled threshold of $n_i\zeta_i$, \ie~$C_2\colon n_i\zeta_i\leq\eta_i$, and $\rho_i$ be the accumulated profit. Let $\tilde{\eta}_i$ be the indicator whether $C_2$ occurs such that $\tilde{\eta}_i=\eta_i$ if $C_2$ occurs and $\tilde{\eta}_i=0$ if $C_1$ occurs. In the $(i+1)$-th iteration, we set $\eta_{i+1}=(\varepsilon\rho_i-2\sum_{j=1}^{i}\tilde{\eta}_j-2)/2$ as long as $\varepsilon\rho_i\ge 2\sum_{j=1}^{i}\tilde{\eta}_j+2$,  which ensures that $2\sum_{j=1}^{i+1}\tilde{\eta}_j +2 \le \varepsilon\rho_i$. With this dynamic strategy, \ATPS could achieve an approximation ratio of $(1-\varepsilon)/3$.

\spara{Time Complexity} As indicated in Algorithm~\ref{alg:AdaptAddiError}, there are $O(\frac{1}{\zeta^2_i}\ln\frac{1}{\delta_i})$ random RR sets generated in the $i$-th iteration. There are at most $2\left \lceil \log(n_i\zeta_0) \right \rceil$ rounds for the $i$-th iteration and $\zeta_0 \in [1/n, 1]$, thus $\delta_i$ is bounded by $O(\frac{1}{kn^3})$. Then the total number of RR sets is at most $O(\frac{\ln n}{\zeta^2_i})$. According to Lemma 4 in~\cite{Tang_TIM_2014}, the expected time for generating a random RR set on $G_i$, denoted as $\mathrm{EPT}$, is $\mathrm{EPT}\leq\frac{m_i}{n_i}\E[I_{G_i}(\{v_i^\circ\})]$ where $v_i^\circ$ is the node with the largest expected spread on $G_i$. By Wald's equation~\cite{Wald_equation_1947}, the expected time complexity of \ATPS is
\begin{equation*}
O\Big(\sum_{i=1}^{k}\big(\frac{\ln n}{\zeta^2_i}\cdot \frac{m_i\E[I(\{v_i^\circ\})]}{n_i}\big)\Big)=O\Big(kmn\E[I(\{v_1^\circ\})]\ln n\Big),
\end{equation*}
since $n_i\zeta_i\geq \sqrt{2}/2$ and $\E[I(\{v_i^\circ\})]\leq \E[I(\{v_1^\circ\})]$ for any $i$.

\begin{theorem}\label{thm:expectetimeadd}
The expected time complexity of\ATPS is $O(kmn\E[I(\{v^\circ\})]\ln n)$ where $v^\circ$ is the node with the largest expected spread on $G$.
\end{theorem}
\section{Optimization with Hybrid Error}\label{sec:adapterror}

In this section, we aim to optimize \ATPS and propose \ATPA in terms of efficiency. We first analyze the rationale of optimization and then conduct theoretical analysis on \ATPA towards its approximation guarantee and time complexity.

\subsection{Rationale of Optimization}\label{sec:reasonforoptimization}

\ATPS in Section~\ref{sec:sec3} considers the additive error on spread estimation, which may suffer from efficiency issues. Recall that for nodes with expected spreads close to the judgement bar, \ATPS uses the stopping condition $n_i\zeta_i \leq 1$ to avoid unnecessary computation overhead. However, when $\zeta_i=O(1/n_i)$, the number of RR sets required is $O(n_i^2\ln n)$, which would incur prohibitive computation overhead. To tackle this issue, we propose to estimate the expected spread via a {\it hybrid error}, \ie the combination of the additive error and relative error. The rationale of hybrid error is that for those nodes with large expected marginal spread, their estimation errors are easily to be bounded within relative error, while for those nodes with small expected marginal spread, their estimation errors are easily to be bounded within additive error. Thus, the expected marginal spread for every node can be efficiently estimated utilizing hybrid error. This estimation reliability using hybrid error is guaranteed by {\it Relative+Additive Concentration Bound} as follows.

\begin{algorithm}[!t]
	\setlength{\hsize}{0.956\linewidth}
	\KwIn{Social graph $G$, a size-$k$ set $T$, the initial error $\varepsilon_0$ and $\zeta_0$, the threshold $\varepsilon$}
	\KwOut{Selected seed node set $S_k$}
	\caption{{\ATPA} \label{alg:AdaptRelAddErr}}	
	Initialize $S_0\leftarrow\emptyset$, $T_0\leftarrow T$\;
	\For{$i\leftarrow 1$ \KwTo $k$}
	{
		\If{$u_i$ is activated}
		{$T_i\leftarrow T_{i-1}\setminus\{u_i\}$, $S_i\leftarrow S_{i-1}$, $G_{i+1}\leftarrow G_i$\; 
		\KwContinue\;}
		$\varepsilon_i \leftarrow \varepsilon_0$, $\zeta_i \leftarrow \zeta_0$, $\delta_i \leftarrow 1/(kn)$\tcp*{$\varepsilon_0\geq \varepsilon,\zeta_0\geq 1/n$}
		\While{True}
		{
			$\theta\leftarrow\frac{(1+\varepsilon_i/3)^2}{2\varepsilon_i\zeta_i}\ln(\frac{4}{\delta_i})$\;
			Generate $\theta$ RR sets as $\R_1$ and $\R_2$, respectively\; 
			$f_\mathrm{est}\leftarrow \Cov_{\R_1}(u_i\mid S_{i-1})\cdot \frac{n_i}{\theta}$\;
			$r_\mathrm{est}\leftarrow\Cov_{\R_2}(u_i\mid T_{i-1} \setminus \{u_i\}) \cdot \frac{n_i}{\theta}$\;
			\If{$C_1^\prime$ \KwOr $C_2^\prime$\label{alg3:judement}}
			{
				\If{$f_\mathrm{est}+r_\mathrm{est} \ge 2c(u_i)$\label{alg3:estimation}}
				{
					$S_i\leftarrow S_{i-1}\cup \{u_i\}$, $T_i \leftarrow T_{i-1}$\;
					Observe the node set $A(u_i)$ activated by $u_i$\;
					Update $G_i$ into $G_{i+1}$ by removing $A(u_i)$\;						
				}
				\mbox{\lElse{$T_i\leftarrow T_{i-1}\setminus\{u_i\}$, $S_i\leftarrow S_{i-1}$, $G_{i+1}\leftarrow G_i$}}\\
				\KwBreak;
			}
			\lIf{$\varepsilon_i <= \varepsilon_t$ \KwAnd $n_i\zeta_i>1$}{$\zeta_i\leftarrow\zeta_i/2$\label{HATP:additive}}
			\lElseIf{$\varepsilon_i>\varepsilon$ \KwAnd $n_i\zeta_i<=1$}{$\varepsilon_i\leftarrow\varepsilon_i/2$\label{HATP:relative}}
			\lElseIf{$f_\mathrm{est}\ge 10 n_i\zeta_i$}{$\varepsilon_i\leftarrow\varepsilon_i/2$}
			\lElseIf{$f_\mathrm{est} <=n_i\zeta_i$}{$\zeta_i\leftarrow\zeta_i/2$}			    			    
			\lElse{$\varepsilon_i\leftarrow\varepsilon_i/\sqrt{2}$, $\zeta_i\leftarrow\zeta_i/\sqrt{2}$\label{HATP:both}}
			$\delta_i\leftarrow\delta_i/2$\;				
		}
	}
	\Return $S_k$\;
\end{algorithm}

\begin{lemma}[Relative+Additive Concentration Bound]\label{lem:RelatAddi}
Let $X_1-\E[X_1], \cdots, X_\theta-\E[X_\theta]$ be a martingale difference sequence such that $X_i \in [0,1]$ for each $i$. Let $X=\frac{1}{\theta}\sum_{i=1}^{\theta}X_i$ and $\mu=\E[X]$. Given $\varepsilon, \zeta \in (0,1)$, then
\begin{align}
&\Pr[X\geq(1+\varepsilon)\mu+\zeta]\le \e^{-\frac{2\theta\varepsilon\zeta}{(1+\varepsilon/3)^2}},\label{eqn:upper}\\
&\Pr[X\leq(1-\varepsilon)\mu-\zeta]\le \e^{-2\theta\varepsilon\zeta}.\label{eqn:lower}
\end{align}
\end{lemma}	
\begin{proof}
	According to the martingale concentration \cite{Tang_IMM_2015}, we have $\Pr\big[{X}\leq(1-\varepsilon)\mu-\zeta\big]\leq \e^{-\frac{(\varepsilon\mu+\zeta)^2{\theta}}{2\mu}}\leq \e^{-\frac{(2\sqrt{\varepsilon\mu\zeta})^2{\theta}}{2\mu}}=\e^{-2\varepsilon\zeta {\theta}}$. Similarly, we have $\Pr[{X}\geq(1+\varepsilon)\mu+\zeta]\leq \e^{-h(\lambda)}$, where $h(\lambda)=\frac{(\lambda^2{\theta})}{2(\lambda-\zeta)/\varepsilon+2\lambda/3}$ and $\lambda=\varepsilon \mu+\zeta$. Let
	\begin{equation*}
	\frac{\mathrm{d} h(\lambda)}{\mathrm{d} \lambda}
	=\frac{\big(2\lambda((\lambda-\zeta)/\varepsilon+\lambda/3)-(1/\varepsilon+1/3)\lambda^2\big)\theta}{2\big((\lambda-\zeta)/\varepsilon+\lambda/3\big)^2}\triangleq 0.
	\end{equation*}
	Thus, $h(\lambda)$ achieves its minimum at $\lambda=\frac{2\zeta}{\varepsilon(1/\varepsilon+1/3)}$ such that $h(\lambda)=\frac{2\varepsilon\beta \theta}{(1+\varepsilon/3)^2}$. This completes the proof.
\end{proof}

Based on the hybrid error, we propose \ATPA\footnote{Algorithms with \underline{h}ybrid error for \underline{a}daptive \underline{t}argeted \underline{p}rofit maximization} algorithm, as shown in Algorithm~\ref{alg:AdaptRelAddErr}.  

\subsection{Description of \ATPA}\label{sec:atpadescrip}

The design principle of \ATPA is similar to that of \ATPS except that \ATPA adopts {\it hybrid error} instead of {\it additive error}. The major differences between \ATPA and \ATPS lie in two aspects. First, the stopping conditions $C_1$ and $C_2$ (Line~\ref{alg3:judement}) have been updated accordingly in \ATPA as follows. 
\begin{align*}
&C_1^\prime\colon \big(\tfrac{f_\mathrm{est} + r_\mathrm{est} -2n_i\zeta_i}{1+\varepsilon}\ge 2c(u_i)\big)\vee \big(\tfrac{r_\mathrm{est}-n_i\zeta_i}{1+\varepsilon} \ge c(u_i)\big) \\
&\phantom{C_1^\prime\colon} \vee \big(\tfrac{f_\mathrm{est} +r_\mathrm{est}+2n_i\zeta_i}{1-\varepsilon} \le 2c(u_i)\big) \vee \big(\tfrac{f_\mathrm{est} +n_i\zeta_i}{1-\varepsilon} \le c(u_i)\big),\\
&C_2^\prime\colon (\varepsilon_i \leq \varepsilon) \wedge (n_i\zeta_i \leq 1).
\end{align*}
Second, the error parameters $\varepsilon_i$ and $\zeta_i$ are adjusted adaptively in \ATPA (Lines~\ref{HATP:relative}--\ref{HATP:both}) instead of decreasing with a fixed ratio in \ATPS. Specifically, if the current addition error $n_i\zeta_i$ reaches the threshold or it is the one magnitude smaller than estimated $f_\mathrm{est}$, we can infer that the expected marginal spread of $u_i$ is much larger than the additive error. In such case, the relative error $\varepsilon_i$ is halved (Line~\ref{HATP:relative}). Similarly, if the current relative error $\varepsilon_i$ reaches the threshold or the additive error is larger than estimated $f_\mathrm{est}$, we should halve the additive error $n_i\zeta_i$ (Line~\ref{HATP:additive}). Otherwise, both relative and additive errors are deceased by a factor of $\sqrt{2}$ (Line~\ref{HATP:both}). This adaptive adjustment could boost the efficiency of \ATPA significantly.

\subsection{Theoretical Analysis of \ATPA}\label{sec:approreladd}

\spara{Approximation Guarantee} To derive the approximation guarantee of \ATPA, we need a similar equation like~\eqref{eqn:policyprofit3} to bridge our solution with the optimal solution in each iteration. Toward this end, we have the following lemma.
\begin{lemma}\label{lem:atpacase3}
For the $i$-th iteration of \ATPA, we have
\begin{equation*}\label{eqn:atpaexp3}
\begin{split}
&\E_{\A}\big[\Lambda(\pi^\circ_{[i-1]}) - \Lambda(\pi^\circ_{[i]})\big] -\frac{2(1+\varepsilon c(u_i))}{1-\varepsilon} -\frac{2}{k}  \\
& \le\E_{\A}\big[\Lambda(\pi^\mathrm{f}_{[i]}) - \Lambda(\pi^\mathrm{f}_{[i-1]})+\Lambda(\pi^\mathrm{r}_{[i]}) - \Lambda(\pi^\mathrm{r}_{[i-1]})\big]
\end{split} 
\end{equation*}
where the expectation $\E_{\A}$ is over the internal randomness of \ATPA and $\varepsilon$ is the threshold of relative error.
\end{lemma} 

Lemma~\ref{lem:atpacase3} establishes the relation between policies $\pi^\mathrm{f}, \pi^\mathrm{r}$ and $\pi^\circ$ for each iteration, based on which, we have following theorem on the approximation of \ATPA.
\begin{theorem}\label{thm:finalratio}
\ATPA achieves the expected profit at least $\frac{\Lambda(\pi^{\opt})- 2(k+\varepsilon c(T))/(1-\varepsilon)-2}{3}$ for any $\varepsilon\in(0,1)$, where $c(T)$ is the cost of $T$.
\end{theorem}

\spara{Time Complexity} In the $i$-th iteration, there are $O(\frac{1}{\varepsilon\zeta_i}\ln\frac{1}{\delta_i})$ random RR sets, where $\delta_i=O(\frac{\varepsilon}{n})$. Thus, the expected time complexity of \ATPA is
\begin{equation*}
O\Big(\sum_{i=1}^{k}\big(\frac{\ln\frac{n}{\varepsilon}}{\varepsilon\zeta_i}\cdot \frac{m_i\E[I(\{v_i^\circ\})]}{n_i}\big)\Big)=O\Big(\frac{km\E[I(\{v_1^\circ\})]}{\varepsilon}\ln\frac{n}{\varepsilon}\Big).
\end{equation*}

\begin{theorem}\label{thm:expectetimeadapt}
The expected time complexity of \ATPA is $O(\frac{km\E[I(\{v^\circ\})]}{\varepsilon}\ln\frac{n}{\varepsilon})$ where $v^\circ$ is the node with the largest expected spread on $G$.
\end{theorem}

Note that \ATPA is approximately $O(\varepsilon n)$ more efficient than \ATPS. Usually, $\varepsilon=O(1)$, e.g., $\varepsilon=0.1$, in the literature, \ATPA achieves a factor of $O(n)$ improvement on efficiency.
\section{Related Work}\label{sec:relate}

As introduced in Section~\ref{sec:doublegreedy}, profit maximization problem is an application of {\it unconstrained submodular maximization (USM)} problem~\cite{FeigeMV_random_2011, Buchbinder_DoubleGreedy_2012}. Thus in this section, we mainly discuss the related work on {\it submodular maximization} and {\it profit maximization}.

\spara{Submodular Maximization} Submodular maximization has attracted considerable interest~\cite{Asadpour_sensor_2008, FeigeMV_random_2011, Buchbinder_DoubleGreedy_2012, Badanidiyuru_fast_2014, Badanidiyuru_adapSeeding_2016} in the past decades. There are a plethora of applications of submodular maximization in real world, such as maximum facility location~\cite{AgeevS_facility_1999}, Max-Cut~\cite{Goemans_maxcut_1995} and influence maximization (IM)~\cite{Kempe_maxInfluence_2003}. Compared with the profit maximization (PM) problem, the IM problem is the most relevant work. Spread function defined in the IM problem is submodular and monotone under the independent cascade (IC) and linear threshold (LT) models~\cite{Kempe_maxInfluence_2003}. However, profit function is submodular but not necessarily monotone, by which profit maximization is {\it unconstrained submodular maximization (USM)}~\cite{FeigeMV_random_2011, Buchbinder_DoubleGreedy_2012}. As pointed out by previous work~\cite{FeigeMV_random_2011, Zhiyi_pricecompete_2018}, there is no efficient approximation algorithms for general USM problem without any additional assumptions. For nonnegative USM, Feige~\etal~\cite{FeigeMV_random_2011} prove that an uniformly random selected method could achieve an $\frac{1}{4}$-approximation (resp.\ $\frac{1}{2}$-approximation) if the submodular function is nonsymmetric (resp.\ symmetric). As what follows, Buchbinder~\etal~\cite{Buchbinder_DoubleGreedy_2012} propose deterministic double greedy and randomized double greedy methods, achieving $\frac{1}{3}$-approximation and $\frac{1}{2}$-approximation for USM respectively under the assumption that submodular function on the ground set is nonnegative.

\spara{Profit Maximization} The profit maximization (PM) problem has been a hot topic in academia recently. The existing work all focuses on PM problem in the nonadaptive setting, \ie nonadaptive PM. Tong~\etal~\cite{Amo_CouponAdvertise_2018} consider the coupon allocation in the profit maximization problem. By utilizing the randomized double greedy~\cite{Buchbinder_DoubleGreedy_2012}, they design algorithms to address the proposed {\it simulation-based profit maximization} and {\it realization-based profit maximization} and claim to achieve $\frac{1}{2}$-approximation with high probability. Liu~\etal~\cite{LiuLWFDW_ProfitCoupon_2018} also consider the coupon allocation in profit maximization under a new diffusion model named {\it independent cascade model with coupons and valuations}. To address this problem, they propose PMCA algorithm based on the local search algorithm~\cite{FeigeMV_random_2011}. PMCA is claimed to achieve an $\frac{1}{3}$-approximation upon the assumption that the submodular function is nonnegative for every subset. However, this assumption is too stringent. Moreover, the time complexity of PMCA is as large as $O(\log(n)mn^4/\varepsilon^3)$, due to which PMCA does not work in practice. Tang~\etal~\cite{Tang_Profit_2018} utilize the deterministic and randomized double greedy algorithms~\cite{Buchbinder_DoubleGreedy_2012} to address the profit maximization problem. With the assumption that submodular function on the ground set is nonnegative, they prove the $\frac{1}{3}$- and $\frac{1}{2}$-approximation guarantees respectively. Furthermore, they design an novel method and relax this assumption to a much weaker one. However, they do not analyze the sampling errors in spread estimation, which makes the proposed algorithms heuristic.

\section{Experiments}\label{sec:experiment}

In this section, we evaluate the performance of our proposed algorithms through extensive experiments. We measure the efficiency and effectiveness in real online social networks. Our experiments are deployed on a Linux machine with an Inter Xeon 2.6GHz CPU and 64GB RAM.

\begin{table}[!t]
	\centering
	\caption{Dataset details. ($\boldsymbol{\textrm{K}=10^3, \textrm{M}=10^6}$)} \label{tbl:dataset}
	\setlength{\tabcolsep}{0.5em} % for the horizontal padding
	\renewcommand{\arraystretch}{1.2}% for the vertical padding
	\begin{tabular} {l|rrrc}
		\hline
		{\bf Dataset} & \multicolumn{1}{c}{$\boldsymbol{n}$} & \multicolumn{1}{c}{$\boldsymbol{m}$} & \multicolumn{1}{c}{\bf Type}  & {\bf Avg.\ deg}  \\ \hline
		{NetHEPT}       & 15.2K			&  31.4K		& 	undirected		 &	4.18       \\ %\hline			
		{Epinions}		 & 132K			&  841K			&  	directed		 &	13.4       \\ %\hline
		{DBLP}			 & 655K			&  1.99M		&  	undirected		 &	6.08       \\ %\hline
		{LiveJournal}   & 4.85M			&  69.0M		&  	directed		 &	28.5     \\ \hline
	\end{tabular}
\end{table}

\subsection{Experimental Setting}\label{sec:expset}

\spara{Datasets} Four online social networks are used in our experiments, namely {NetHEPT}, {Epinions}, {DBLP}, {LiveJournal}, as presented in Table~\ref{tbl:dataset}. Among them, {NetHEPT}~\cite{Chen_NewGreedy_2009} is the academic collaboration networks of ``High Energy Physics-Theory'' area. The rest three datasets are real-life social networks available in~\cite{Leskovec_SNAP_2014}. In particular, {LiveJournal} contains millions of nodes and edges. For fair comparison, we randomly generate 20 possible realizations for each dataset, and report the average performance of each tested algorithm on the 20 possible realizations.

\spara{Algorithms} First, we evaluate the two proposed adaptive algorithms \ATPA and \ATPS. We also adopt the {\it random set (RS)} algorithm~\cite{FeigeMV_random_2011} and extend it into an adaptive version, \ie~{\it adaptive random set} (\ARS). Specifically, \ARS selects each seed node candidate with probability of $0.5$ without reference to its quality. If one node is selected, it then observes and removes all the nodes activated by this node from the graph. (The removed nodes are not examined and selected by \ARS.) This process is repeated until all nodes in the target set have been decided. To verify the advantage of adaptive algorithms over nonadaptive algorithms, we tailor \ATPA into a nonadaptive version, referred to as \NTPA\footnote{Algorithms with \underline{h}ybrid error for \underline{n}onadaptive \underline{t}argeted \underline{p}rofit maximization}, to address the nonadaptive TPM problem. Meanwhile, we also include two extra nonadaptive algorithms proposed in the latest work for nonadaptive profit maximization problem~\cite{Tang_Profit_2018}, \ie~{\it nonadaptive simple greedy} (\NSG) and {\it nonadaptive double greedy} (\NDG). Note that \NSG and \NDG are nonadaptive algorithms where all the seed nodes are selected in one batch before we deploy these nodes into the viral marketing campaign. Also, the analysis in \cite{Tang_Profit_2018} ignores the sampling errors in spread estimation; in contrast, our \ATPS and \ATPA algorithms take such errors into account. For relatively fair comparison, we set the sample size of \NSG and \NDG as the largest number of samples generated in \ATPA for one iteration in all settings. Recall that by evaluating the efficacy of adaptive algorithms \ATPS and \ATPA, we could verify the advantage of our proposed adaptive policies over nonadaptive policies on target profit maximization.

\begin{figure*}[!t]
	\centering
	\includegraphics[height=10pt]{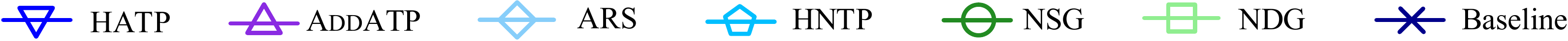}\vspace{-0.15in}\\
	\subfloat[NetHEPT]{\includegraphics[width=0.25\linewidth]{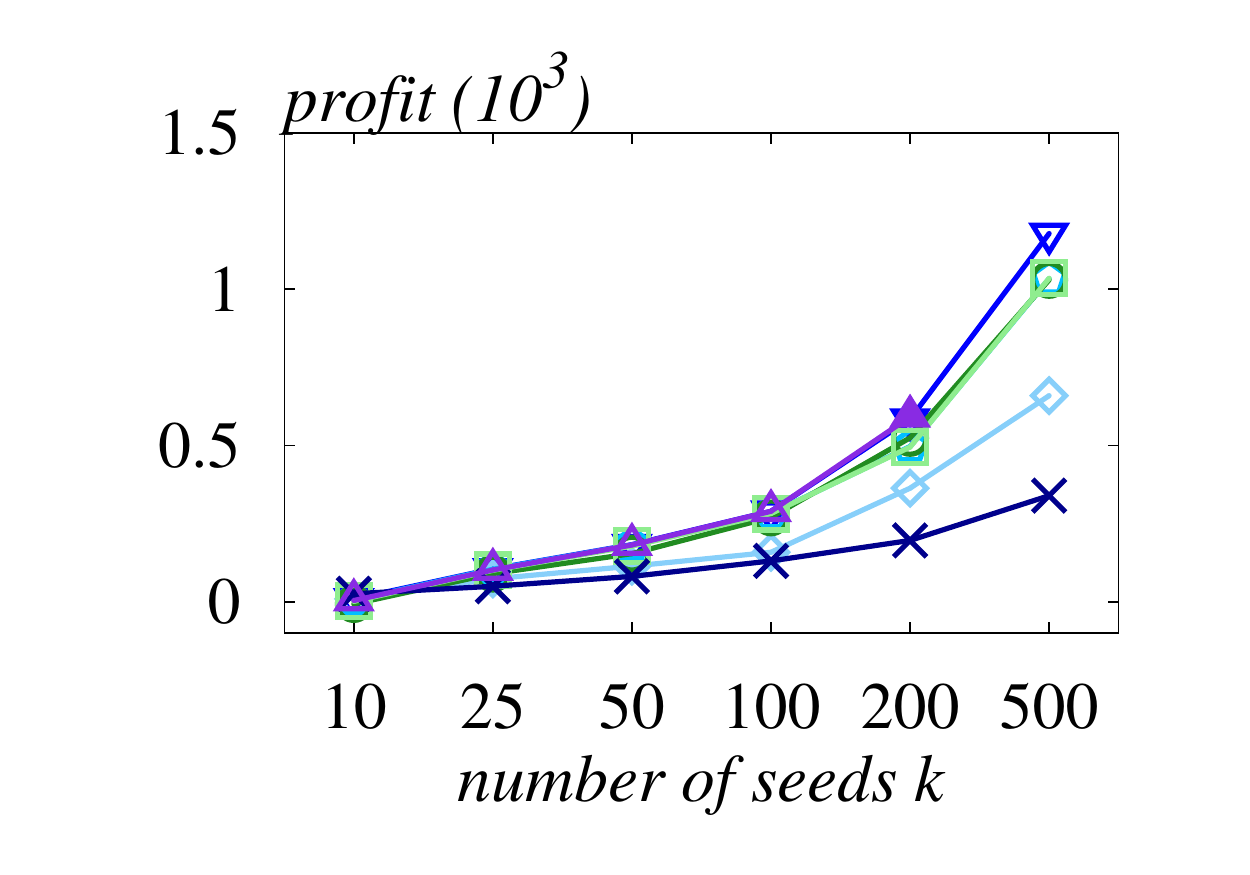}\label{fig:NetHEPT_deg_profit}}\hfill
	\subfloat[Epinions]{\includegraphics[width=0.25\linewidth]{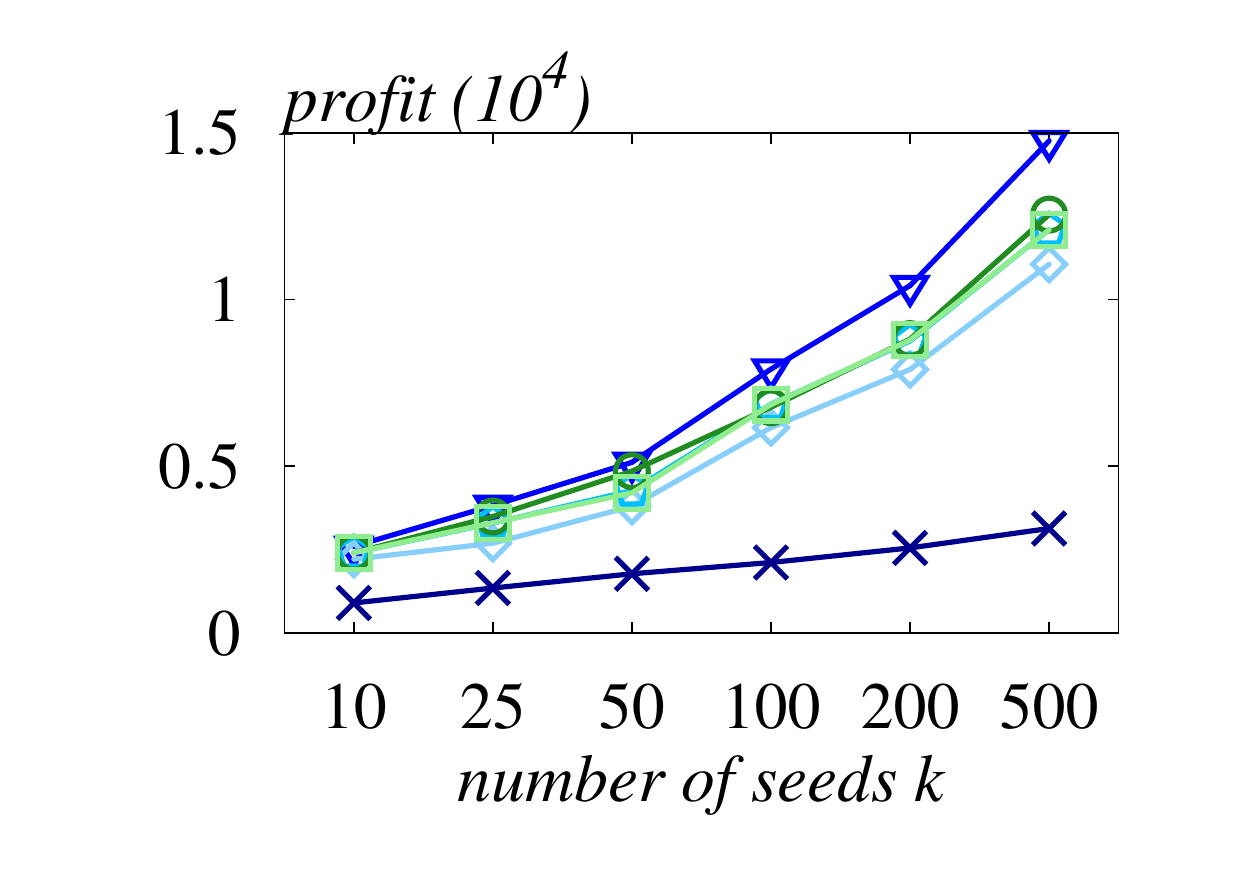}\label{fig:Epinions_deg_profit}}\hfill
	\subfloat[DBLP]{\includegraphics[width=0.24\linewidth]{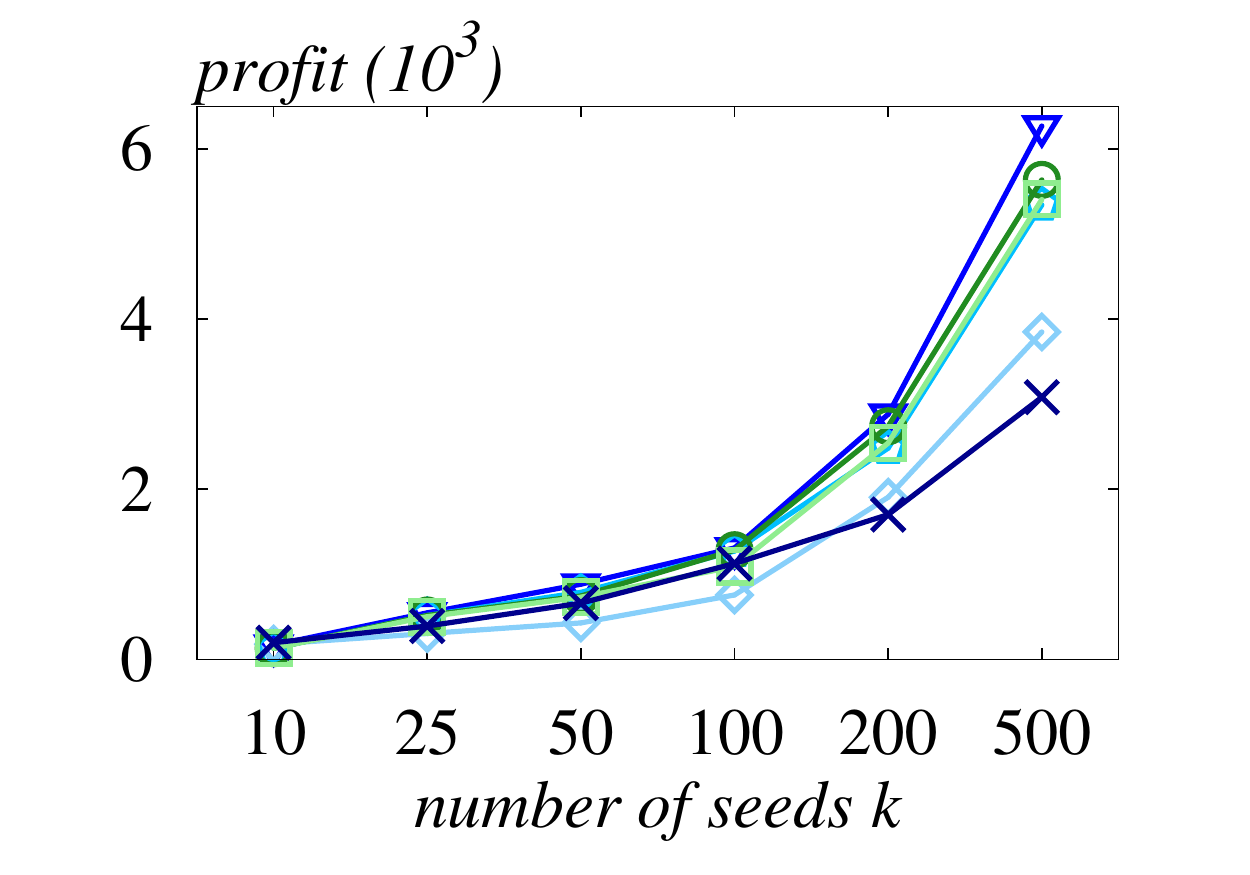}\label{fig:DBLP_deg_profit}}\hfill
	\subfloat[LiveJournal]{\includegraphics[width=0.24\linewidth]{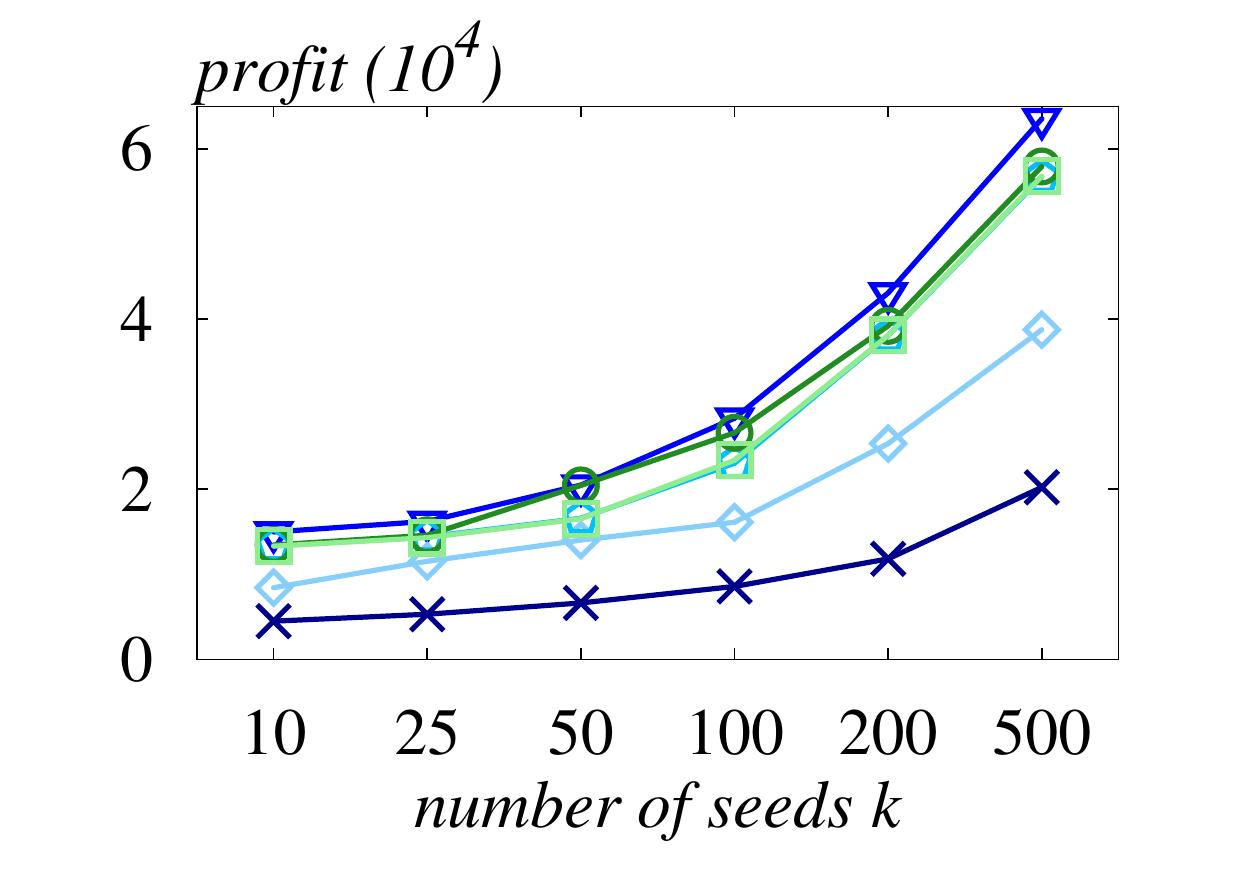}\label{fig:LiveJournal_deg_profit}}
	\caption{Profit in degree-proportional cost.}\label{fig:profit-deg}
\end{figure*}

\begin{figure*}[!t]
	\centering
	\vspace{-0.2in}
	\subfloat[NetHEPT]{\includegraphics[width=0.25\linewidth]{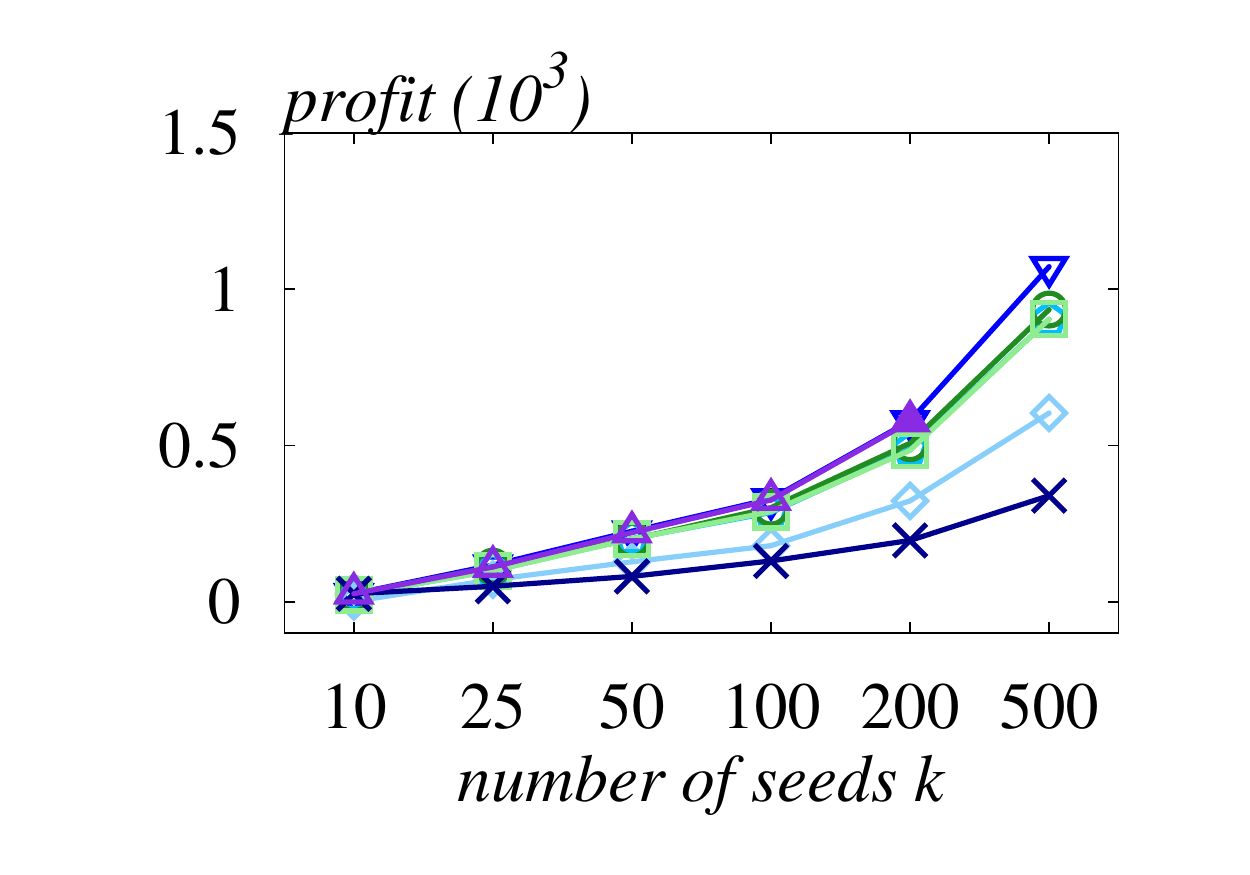}\label{fig:NetHEPT_uni_profit}}\hfill
	\subfloat[Epinions]{\includegraphics[width=0.25\linewidth]{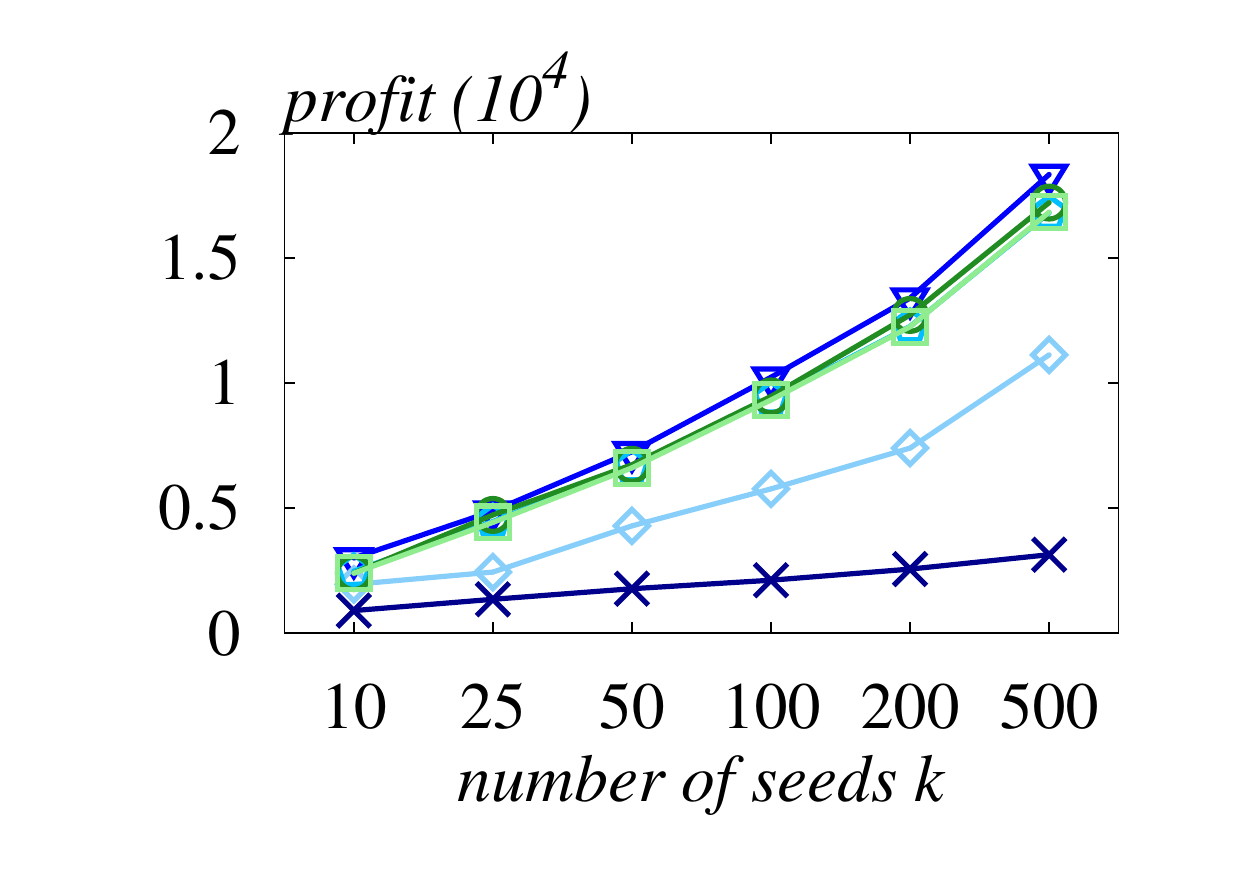}\label{fig:Epinions_uni_profit}}\hfill
	\subfloat[DBLP]{\includegraphics[width=0.24\linewidth]{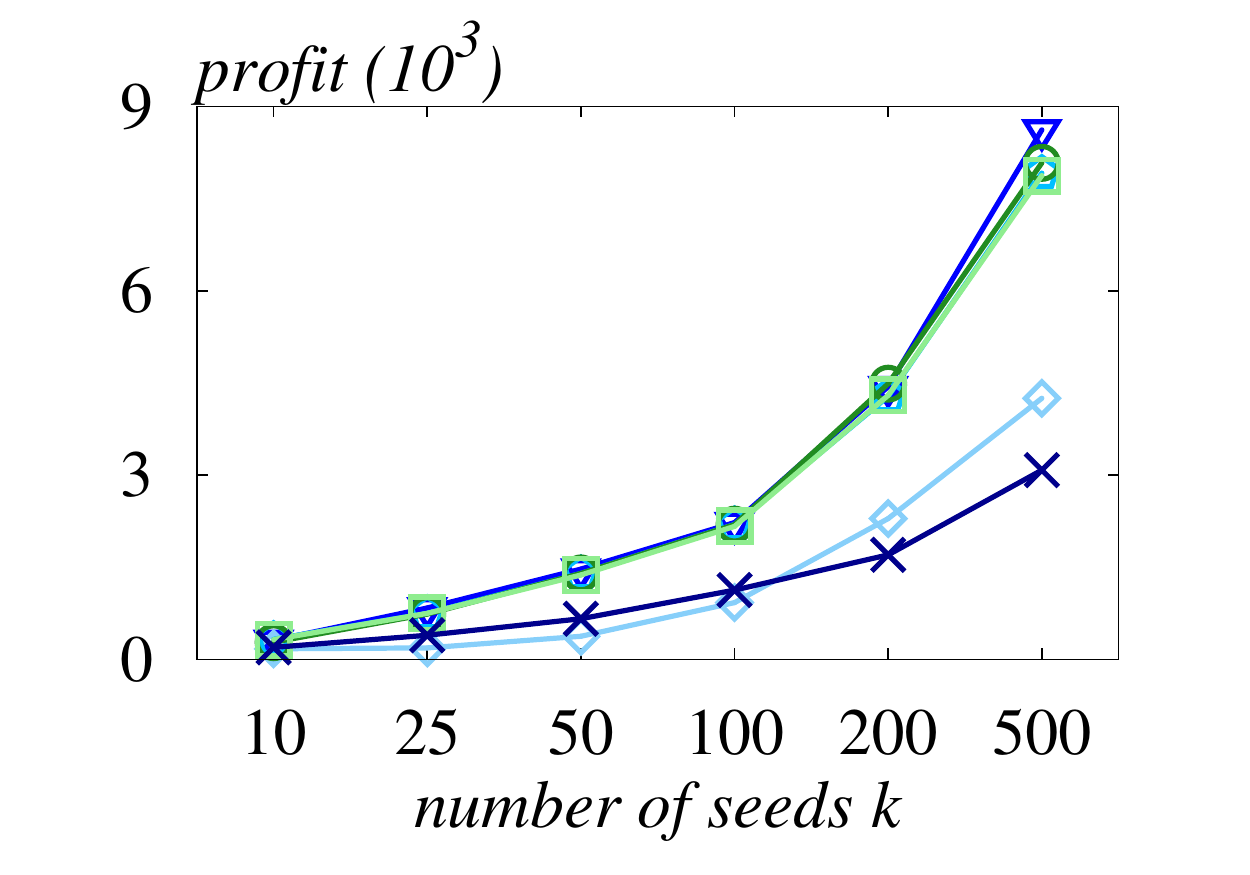}\label{fig:DBLP_uni_profit}}\hfill
	\subfloat[LiveJournal]{\includegraphics[width=0.24\linewidth]{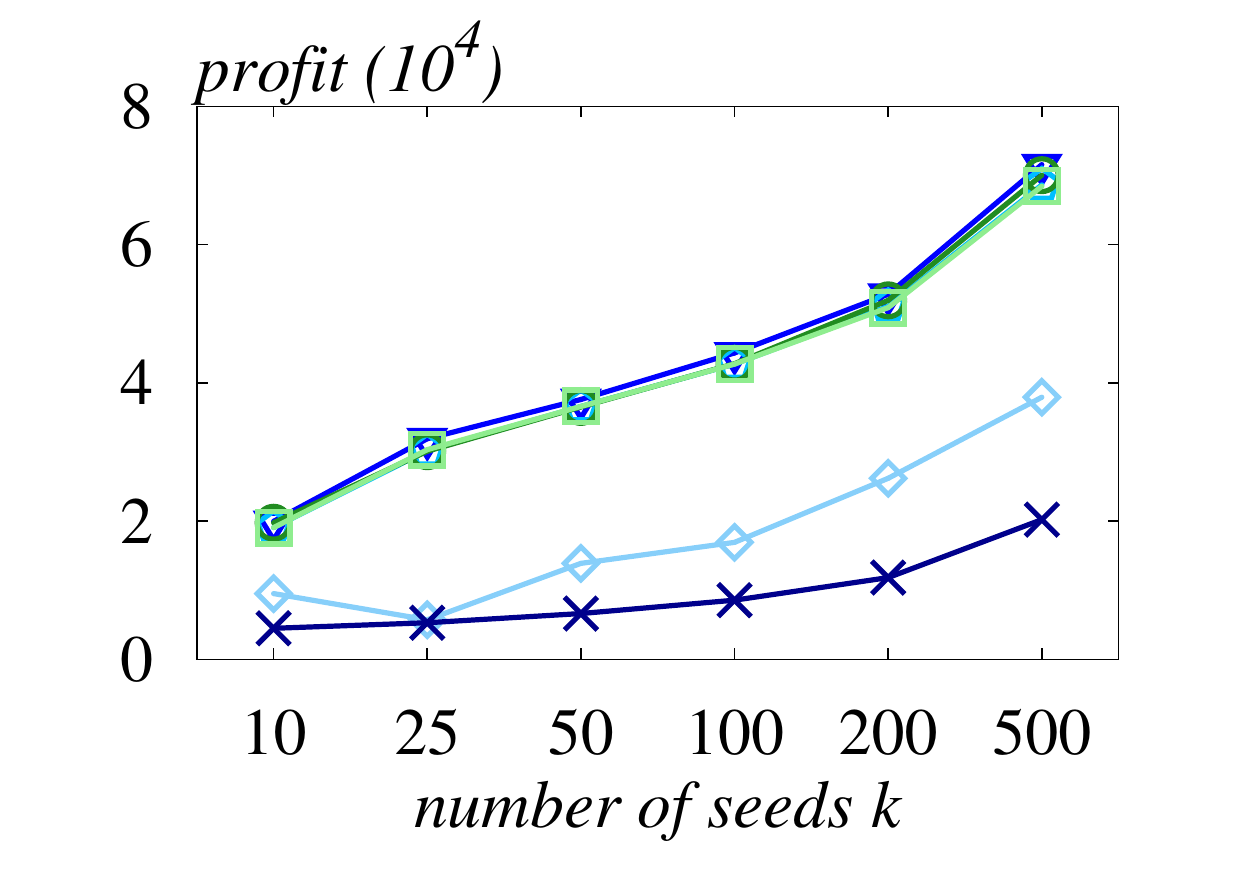}\label{fig:LiveJournal_uni_profit}}
	\caption{Profit in uniform cost.}\label{fig:profit-uni}
\end{figure*} 

\spara{Parameter settings} To conduct a comprehensive evaluation, we design two different procedures to obtain a suitable target set $T$ and the corresponding cost of each user in $T$. First, we follow the setting in~\cite{Arthur_pricestrategy_2009, Lu_profitmax_2012} where the cost of each user in $T$ is based on the expected spread of $T$. We use one of the state of the arts~\cite{Tang_IMM_2015} for influence maximization to obtain the top-$k$ influential users as the target seed set. We vary the target size $k$ as $k=\{10, 25, 50, 100, 200, 500\}$. To determine the cost of each user $u \in T$, we (i) estimate the lower bound of $T$'s expected spread $\E[I(T)]$ as $\E^l[I(T)]$, and (ii) ensure that $c(T)=\E^l[I(T)]$. Under this condition, we design two cost settings, \ie~{\it degree-proportional} cost setting and {\it uniform} cost setting. In the degree-proportional setting, the cost of each node is proportional to its out-degree. In the uniform cost setting, the cost of each node is equal.

Second, we follow the setting in the latest work~\cite{Tang_Profit_2018} where the cost of each user in the graph $G$ is predefined before we choose the target set $T$. Let $\lambda=c(V)/n$ be the ratio of cost to node number where $n=|V|$ and $c(V)$ is the cost of $V$. To get a proper size of $T$, we vary $\lambda$ as $\lambda=\{200, 300, 400, 500\}$. Then the cost of each node in $V$ also follows the degree-proportional cost setting and uniform cost setting respectively. We then adopt the two proposed methods, namely \NDG and \NSG to identify the target set $T$. 

In a nutshell, in the first setting, we choose a target set $T$ first and then assign a cost to each node in $T$ based on the expected spread of $T$, while in the second setting, we assign a cost to each node first and then find a target set $T$ based on the cost assignment. For the two settings, we set the input parameters $n_i\zeta_0=64$, $\varepsilon_0=0.5$, and its threshold $\varepsilon=0.05$ in \ATPA and \NTPA. Following the common setting in the literature~\cite{Tang_TIM_2014,Huang_SSA_2017,Tang_IMhop_2018,Tang_infMax_2017,Tang_OPIM_2018}: for each dataset, we set the edge probability $p(\langle u,v\rangle)=\frac{1}{\indeg_v}$, where $\indeg_v$ is the in-degree of node $v$.

\subsection{Comparison of Profit}\label{sec:profit}

\spara{Degree-proportional Cost} \figurename~\ref{fig:profit-deg} reports the profits achieved by the six tested algorithms under the degree-proportional cost setting. In addition, the dark-blue line with cross mark ({\bf $\pmb{\times}$}) named {\it Baseline} represents the estimated profit of the target set $T$. We observe that all six tested algorithms could improve the profit of baseline significantly. In particular, \ATPA shows superior advantage over the other three nonadaptive algorithms. Specifically, profit achieved by \ATPA is around $10\%$--$15\%$ larger than those of nonadaptive algorithms on average, where the improvement percentage is as high as $20\%$ on the {Epinions} dataset. This verifies the effectiveness of our solutions. Meanwhile, \ATPS achieves comparable profit with \ATPA on the NetHEPT dataset. However, \ATPS runs {\it out of memory} on other larger datasets. (The filled triangle ($\blacktriangle$) represents the largest value of $k$ that \ATPS can run.)

As with the nonadaptive algorithms, \ie \NTPA, \NSG, and \NDG, they obtain quite comparable profits on the four datasets. In particular, we observe that the profit of \NSG is slightly higher than the profits of \NTPA and \NDG on datasets {DBLP} and {LiveJournal}, which implies the minor advantage of simple greedy over double greedy on large datasets. We also observe that \ARS achieves the lowest profits among the six algorithms, since \ARS selects each node with probability of $0.5$ without reference to its quality.

\spara{Uniform Cost} \figurename~\ref{fig:profit-uni} shows the profit results under the uniform cost setting. At the first glance, the results follow the trend of \figurename~\ref{fig:profit-deg}. However, there are two major distinctions on the profit results between the two cost settings. First, algorithms achieve around $50\%$ more profits in the uniform cost setting than they do in the degree-proportional cost setting. This can be explained as follows. In the degree-proportional cost setting, the cost of each user is proportional to its out-degree which is highly correlated with its expected spread. In this regard, each user's cost is roughly proportional to its spread, which largely limits the profits the influential nodes could contribute in some degree. Contrarily, in the uniform cost setting where all users are assigned with the same cost, influential nodes can be more easily picked out to exert their influence and bring more profits.

The other notable distinction is that the gap of profit between adaptive algorithms and nonadaptive algorithms becomes smaller in the uniform cost setting. As aforementioned, those profitable nodes are easier to be identified under this setting. As a consequence, the overlap of selected seed sets between adaptive and nonadaptive algorithms expands, which weakens the adaptivity advantage slightly.

\begin{figure}[!t]
	\centering
	\subfloat[Profits under the random cost]{\includegraphics[width=0.5\linewidth]{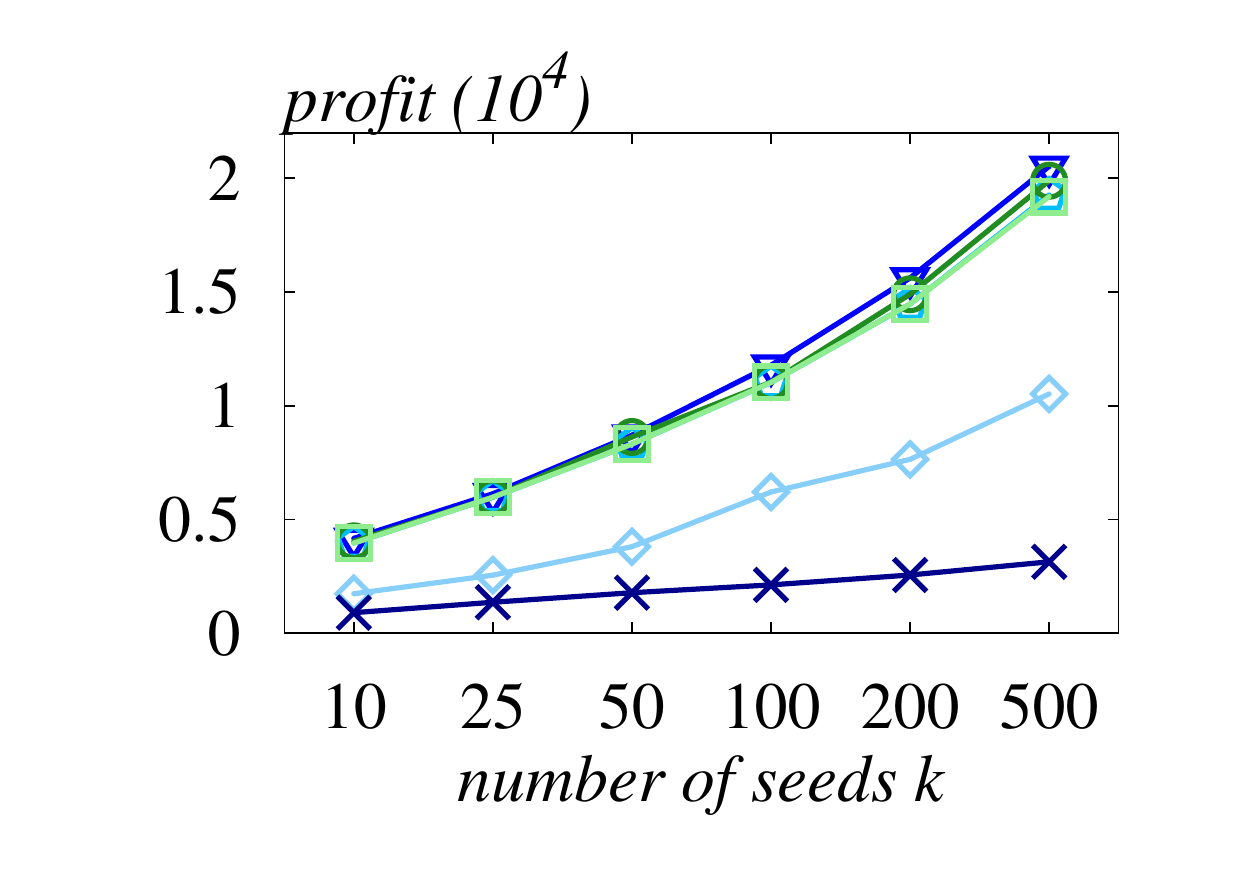}\label{subfig:randcost}}\hfill
	\subfloat[Sensitivity of relative error]{\includegraphics[width=0.5\linewidth]{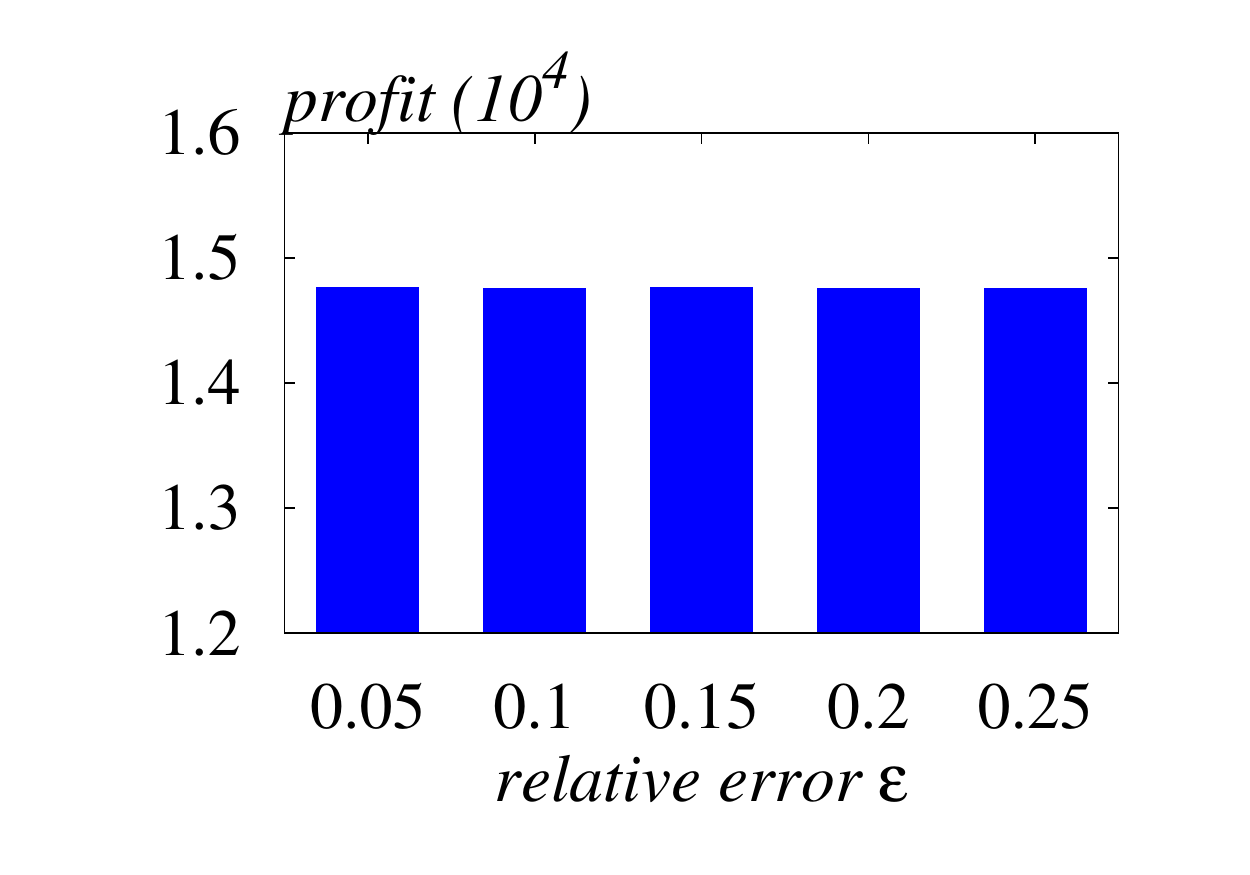}\label{subfig:sensitive}}
	\caption{Profits on {Epinions}.}\label{fig:profit-randsen}
\end{figure}

\begin{figure*}[!t]
	\centering
	\includegraphics[height=11pt]{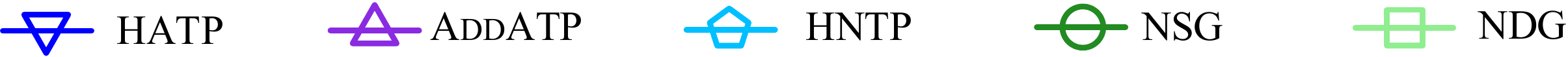}\vspace{-0.15in}\\
	\subfloat[NetHEPT]{\includegraphics[width=0.25\linewidth]{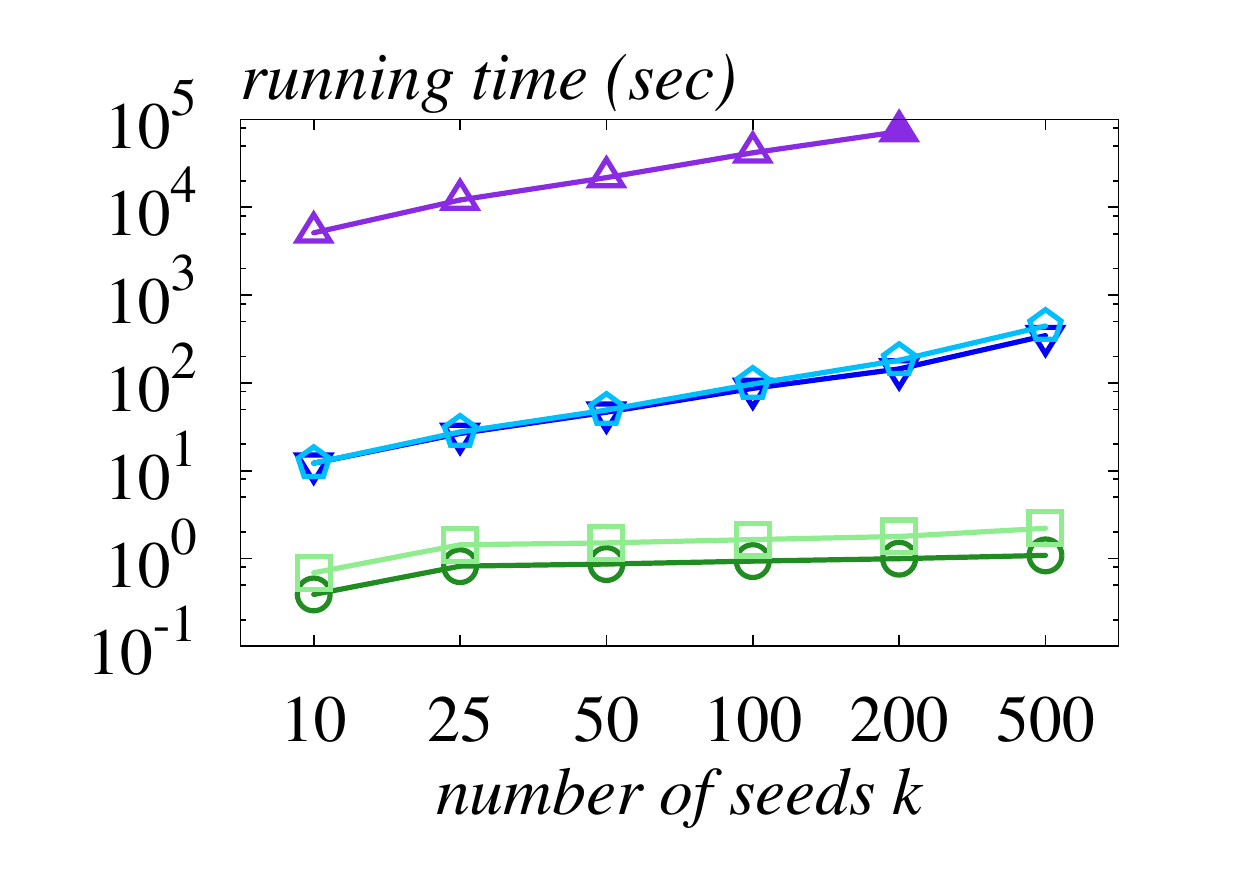}\label{fig:NetHEPT_deg_time}}\hfill
	\subfloat[Epinions]{\includegraphics[width=0.25\linewidth]{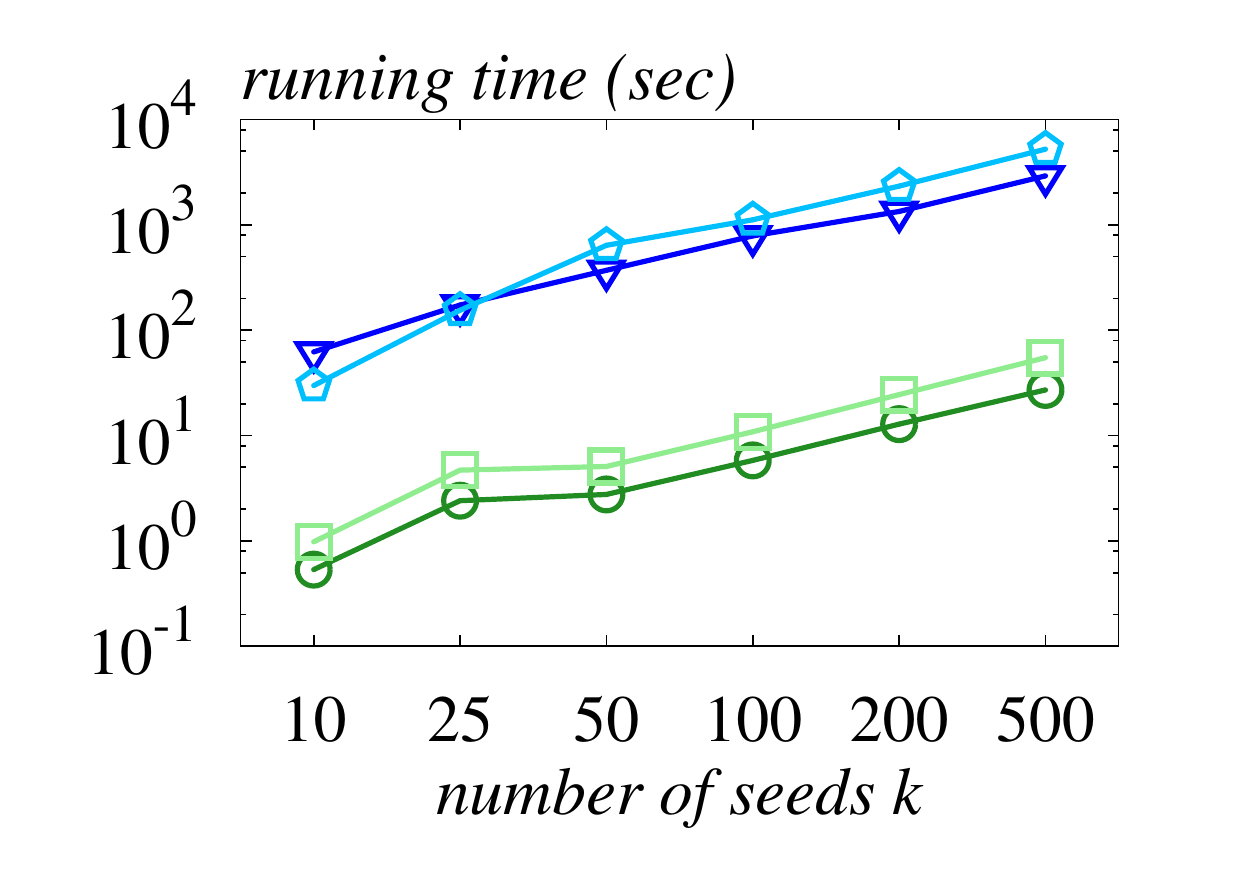}\label{fig:Epinions_deg_time}}\hfill
	\subfloat[DBLP]{\includegraphics[width=0.24\linewidth]{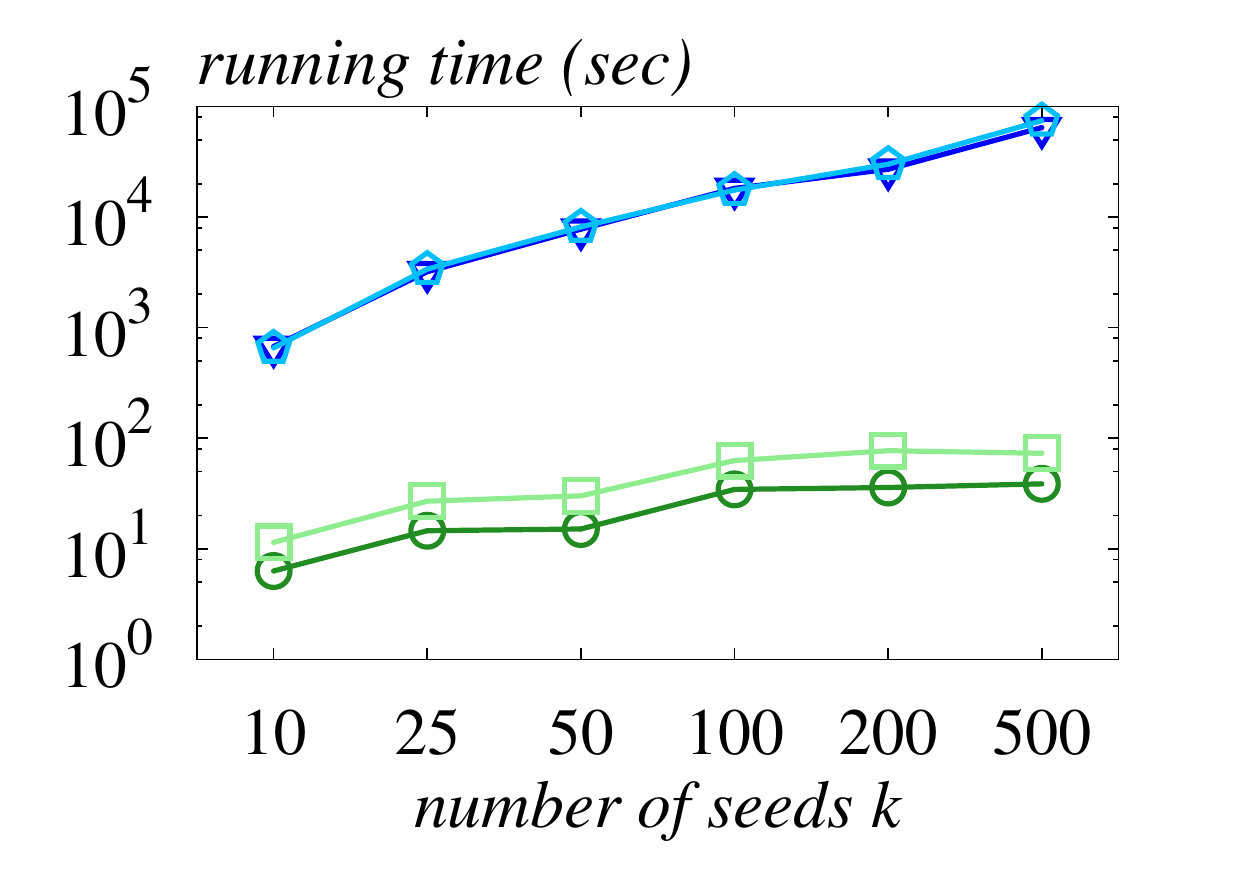}\label{fig:DBLP_deg_time}}\hfill
	\subfloat[LiveJournal]{\includegraphics[width=0.24\linewidth]{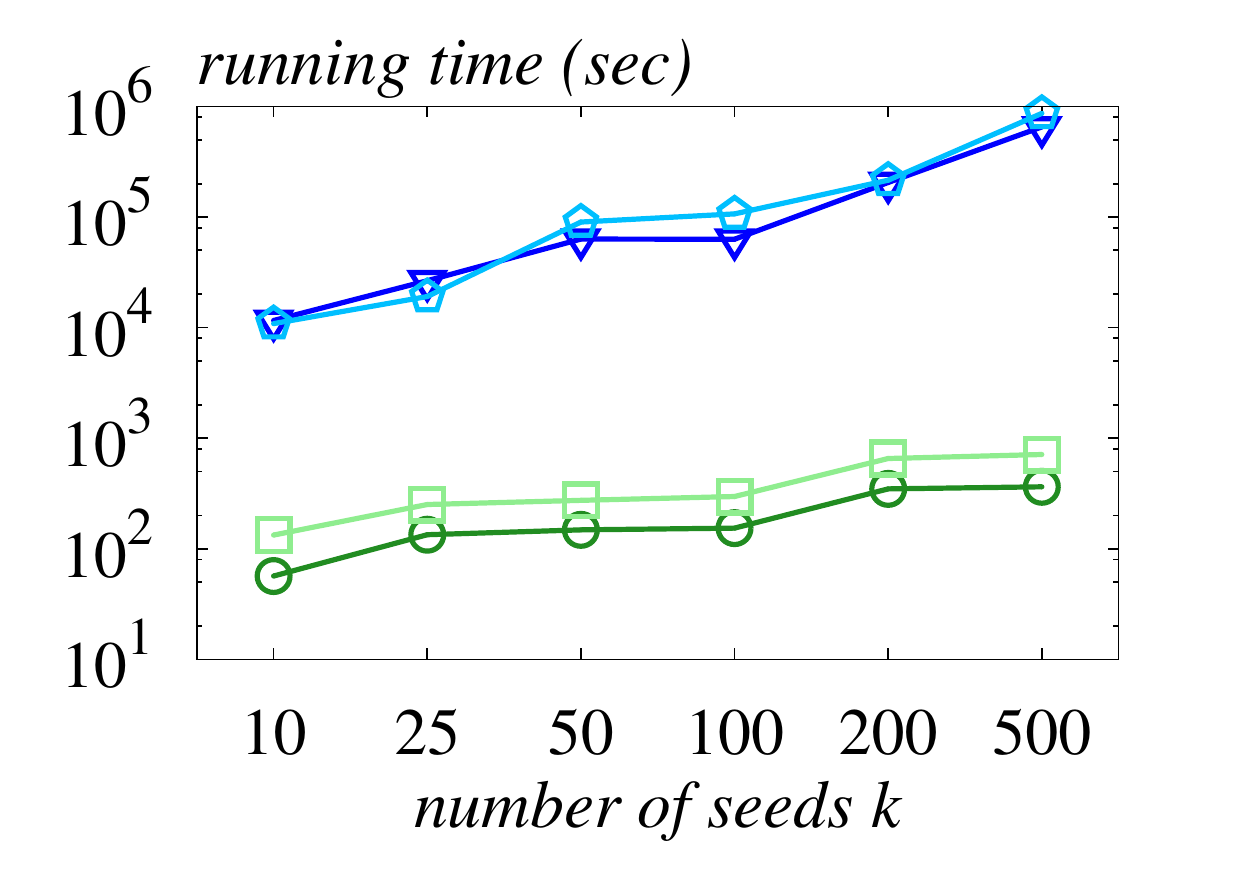}\label{fig:LiveJournal_deg_time}}
	\caption{Running time in degree-proportional cost.}\label{fig:time-deg}
\end{figure*}

\begin{figure*}[!t]
	\centering
	\vspace{-0.2in}
	\subfloat[NetHEPT]{\includegraphics[width=0.25\linewidth]{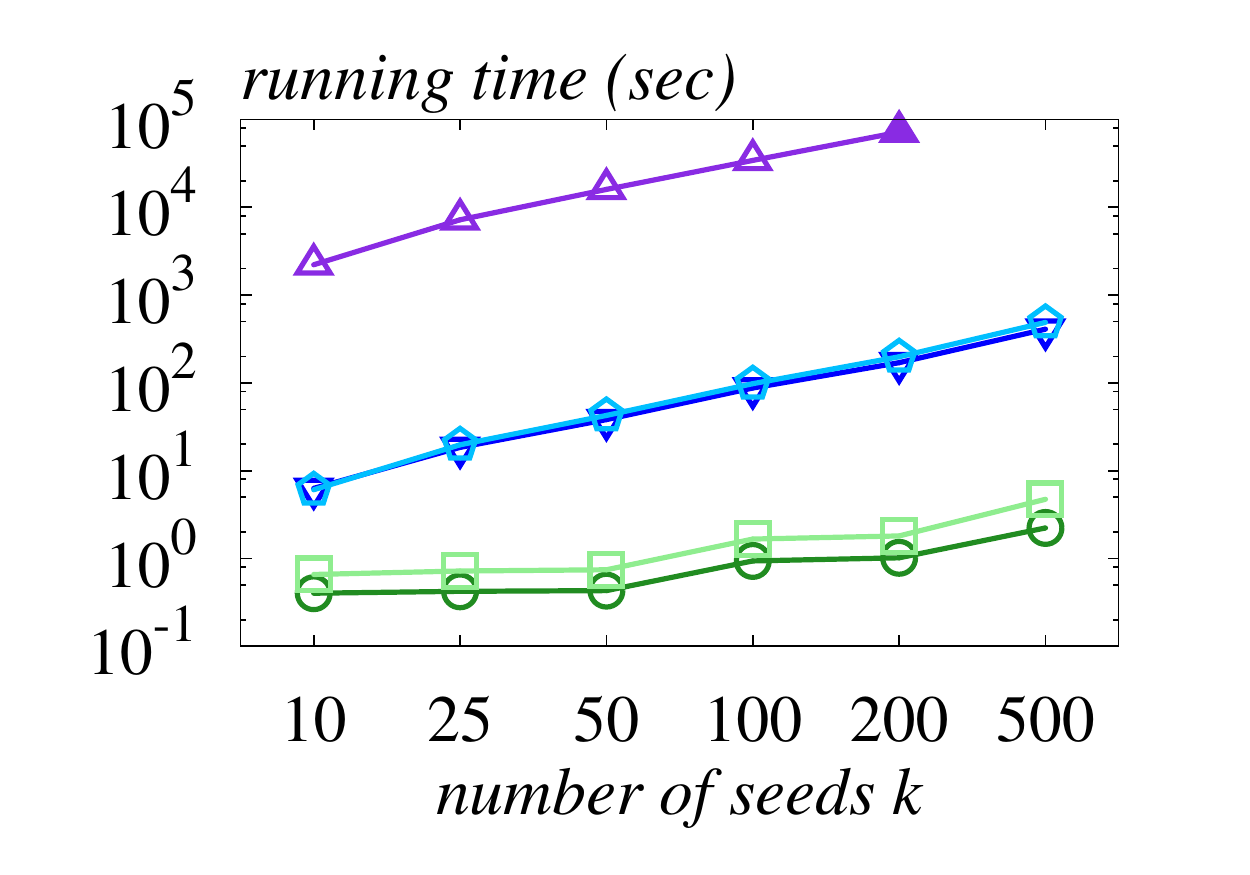}\label{fig:NetHEPT_uni_time}}\hfill
	\subfloat[Epinions]{\includegraphics[width=0.25\linewidth]{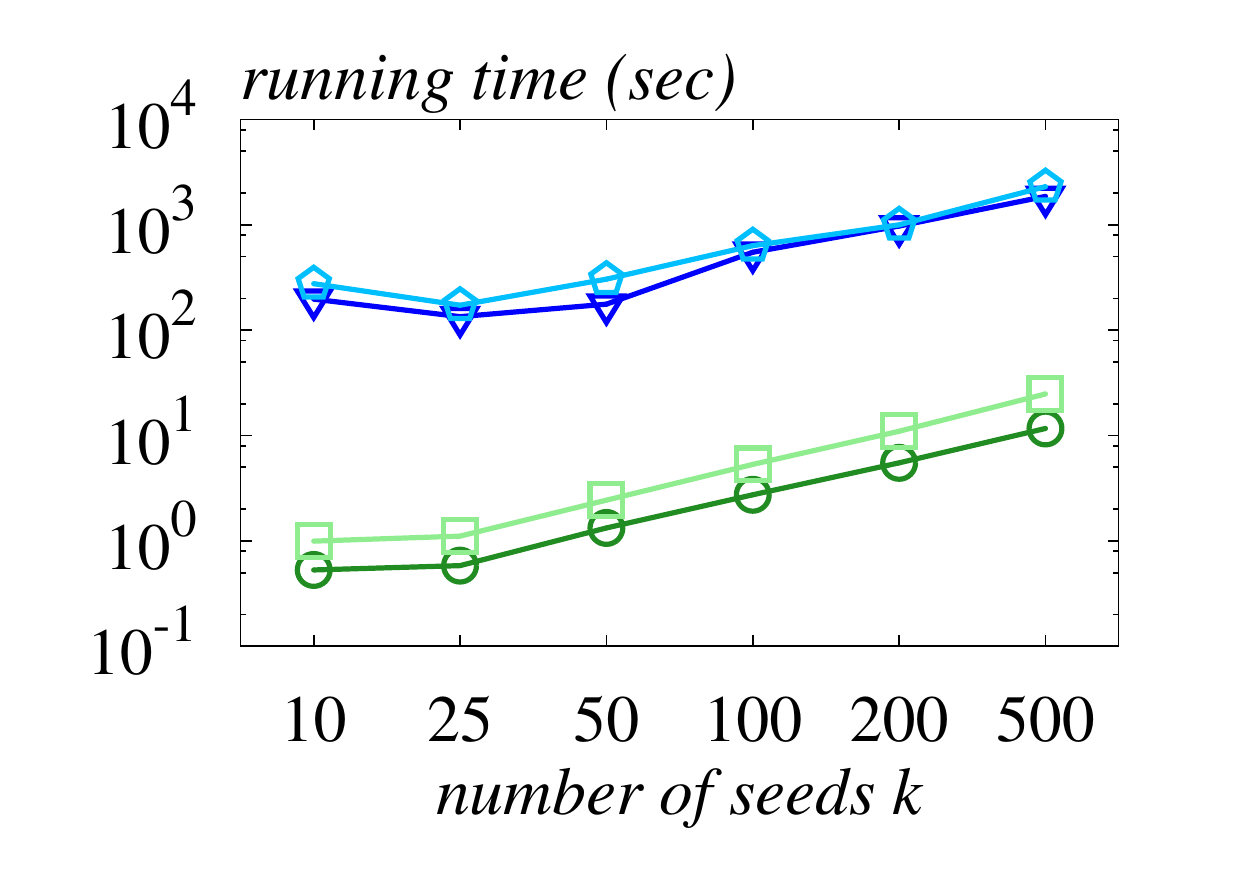}\label{fig:Epinions_uni_time}}\hfill
	\subfloat[DBLP]{\includegraphics[width=0.24\linewidth]{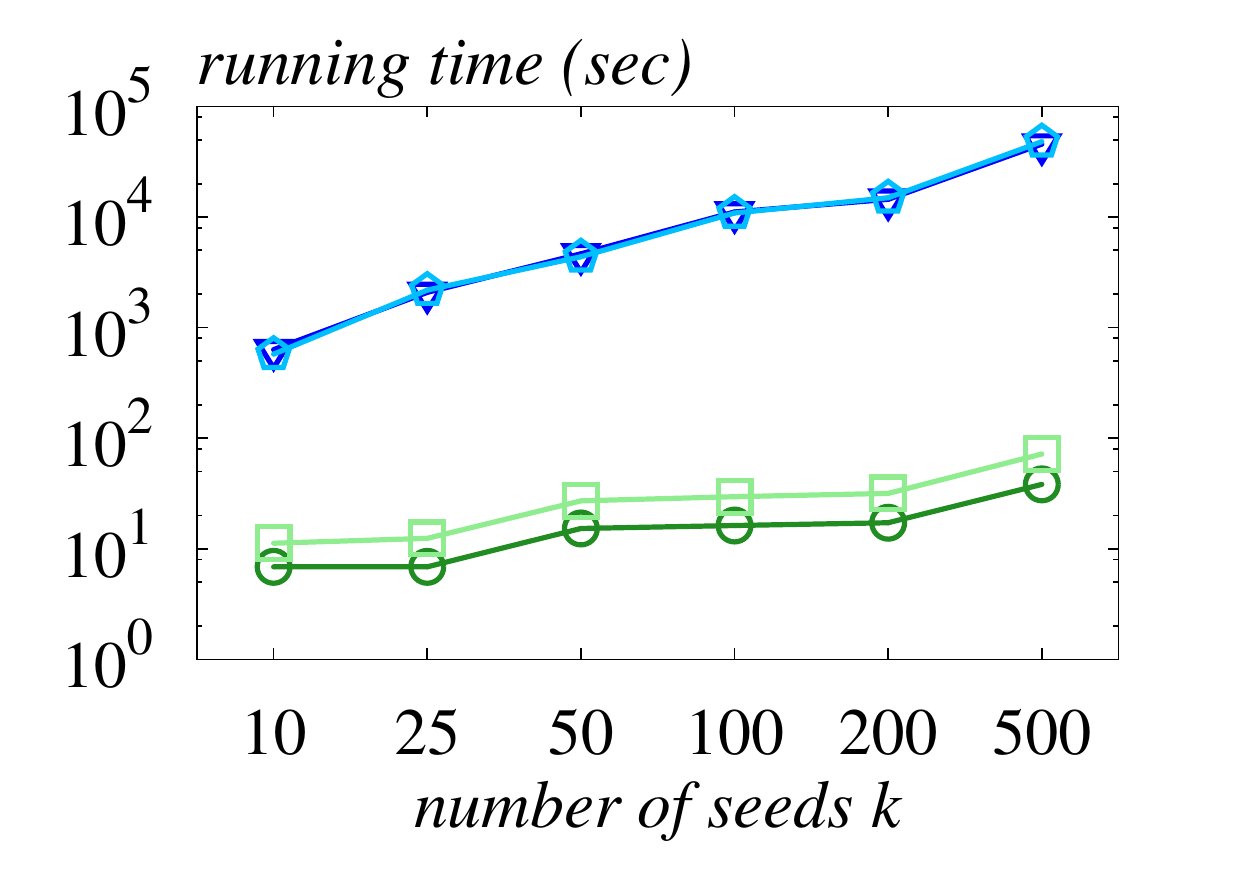}\label{fig:DBLP_uni_time}}\hfill
	\subfloat[LiveJournal]{\includegraphics[width=0.24\linewidth]{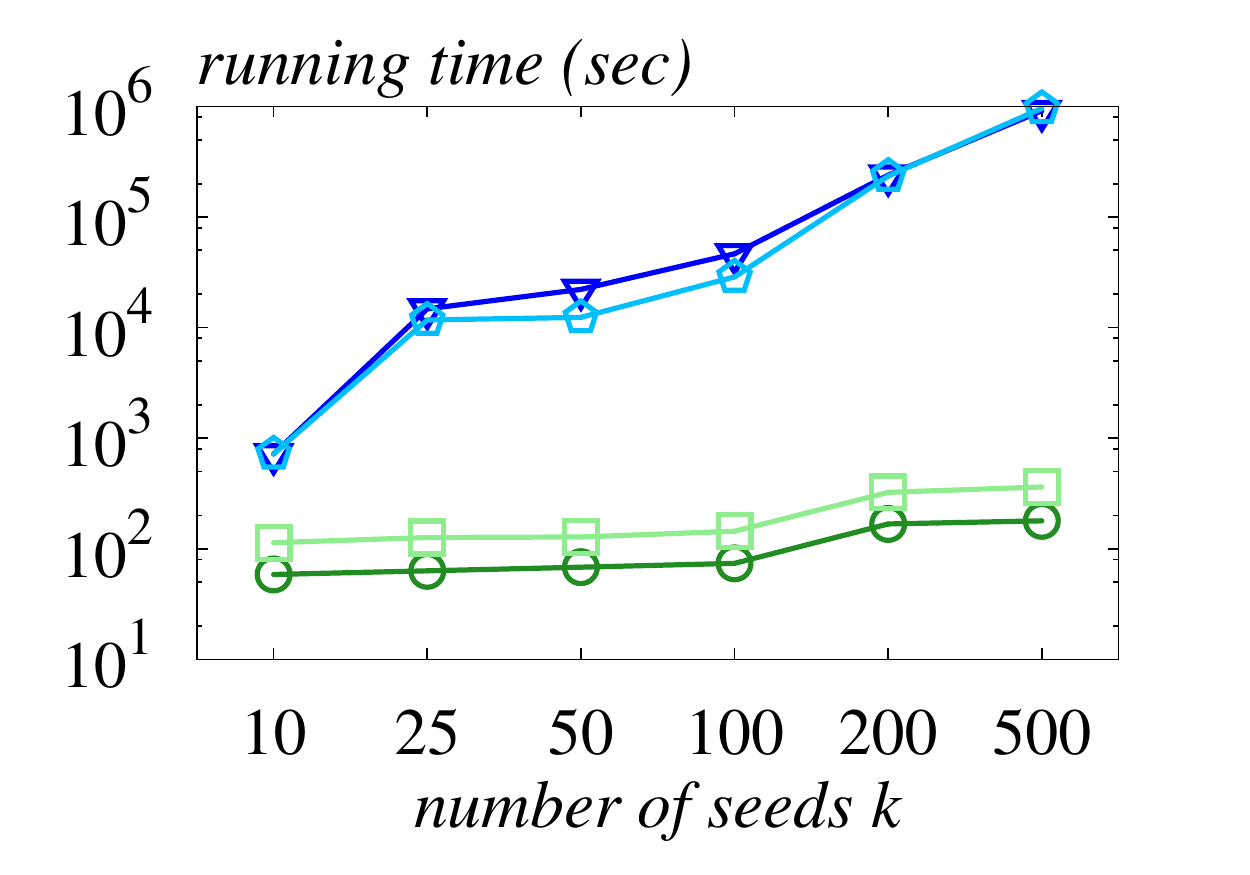}\label{fig:LiveJournal_uni_time}}
	\caption{Running time in uniform cost.}\label{fig:time-uni}
\end{figure*}
\spara{Random Cost} \figurename~\ref{subfig:randcost} shows the results under the {\it random cost} setting where the cost of each node is randomly assigned such that $c(T)=\E^l[I(T)]$. Due to the space limitations, we only present the profit results on dataset {Epinions}. We observe that \ATPA again achieves the highest profits, with around $5\%$ more profits than the other three nonadaptive algorithms. We also observe that (i) the advantage of adaptive algorithms over nonadaptive algorithms becomes less significant under the random cost setting, and (ii) the profits achieved by different algorithms under the random cost setting are around $40\%$ (resp.\ $20\%$) more than those under the degree-proportional (resp.\ uniform) cost setting respectively. The reason is that under the random cost setting, user costs have no correlation with their spreads. In such cases, profitable nodes are easier to be identified by both adaptive and nonadaptive algorithms in expectation, which lessens the adaptivity advantage. Meanwhile, those profitable nodes will influence the same number of nodes but have relatively less costs, which could improve the final profit.

\spara{Sensitivity Test of $\varepsilon$} Recall that \ATPA involves a key parameter of $\varepsilon$, which represents the approximation guarantee of \ATPA (Section~\ref{sec:atpadescrip}). Specifically, we vary the value of $\varepsilon$ as $\{0.05, 0.1, 0.15, 0.2, 0.25\}$ under $k=500$ and the degree-proportional cost setting on the {Epinions} dataset. \figurename~\ref{subfig:sensitive} presents the profit achieved by \ATPA with different $\varepsilon$. As shown, the profits remain nearly steady for all settings, which demonstrates the robustness of \ATPA on the setting of $\varepsilon$.

\subsection{Comparison of Running Time}\label{sec:runningtime}

\spara{Degree-proportional Cost} \figurename~\ref{fig:time-deg} presents the results of running time under the degree-proportional cost setting. Note that in \figurename~\ref{fig:time-deg} and \figurename~\ref{fig:time-uni}, the value of $k$ in x-axis is roughly in exponential scale instead of linear scale. Considering that \ARS selects node randomly without generating any samples, we do not include its running time (which is vary small) in all figures. \figurename~\ref{fig:NetHEPT_deg_time} shows that \ATPS runs around $400$ times slower than \ATPA. This confirms that our optimization techniques for \ATPA can significantly reduce running time. Meanwhile, \ATPA and \NTPA run slower than the heuristic algorithms \NSG and \NDG. This is because \ATPA and \NTPA regenerate RR sets in all $k$ iterations from scratch to ensure bounded sampling errors, as shown in Algorithm~\ref{alg:AdaptRelAddErr}. On the contrary, \NSG and \NDG complete seed selection on one set of RR sets as they do not have any guarantee on profit estimation. Therefore, \ATPA and \NTPA generate around $k$ times samples than \NSG and \NDG do.

Meanwhile, \NTPA runs slightly slower than \ATPA does. Recall that \NTPA is the nonadaptive version of \ATPA and there is no any graph update on $G$ after each node selection. Contrarily, \ATPA would update current graph into residual graph by removing newly activated nodes in each iteration. We observe that generating an RR set on a smaller residual graph is faster than that on the original graph. Therefore, \ATPA spends less time than \NTPA on sampling.

\spara{Uniform Cost} \figurename~\ref{fig:time-uni} displays the results of running time under the uniform cost setting. The result trend keeps consistent with that in degree-proportional cost setting. The major difference between the two is that the running time in \figurename~\ref{fig:time-uni} is smaller than that in \figurename~\ref{fig:time-deg} for each algorithm on the corresponding setting. This can be explained with the reasons aforementioned, \ie profitable nodes in uniform cost setting are easier identified within less samples. Therefore, the corresponding running time becomes shorter.

\subsection{Comparison of Profit with Predefined Cost}\label{sec:profprecost}

This section explores the profit improvement of \ATPA over \NDG and \NSG, following the setting that the cost of each user is predefined before we choose the target set $T$. Specifically, after the cost of each user is set, we adopt \NDG and \NSG to derive the target seed set $T$ respectively. As implied in Section~\ref{sec:profit}, profit improvement exhibits similar characteristics on the four datasets. Therefore, this part of experiment is conducted only on the largest dataset {LiveJournal}. 

\begin{figure}[!t]
	\centering
	\subfloat[Degree-proportional cost]{\includegraphics[width=0.485\linewidth]{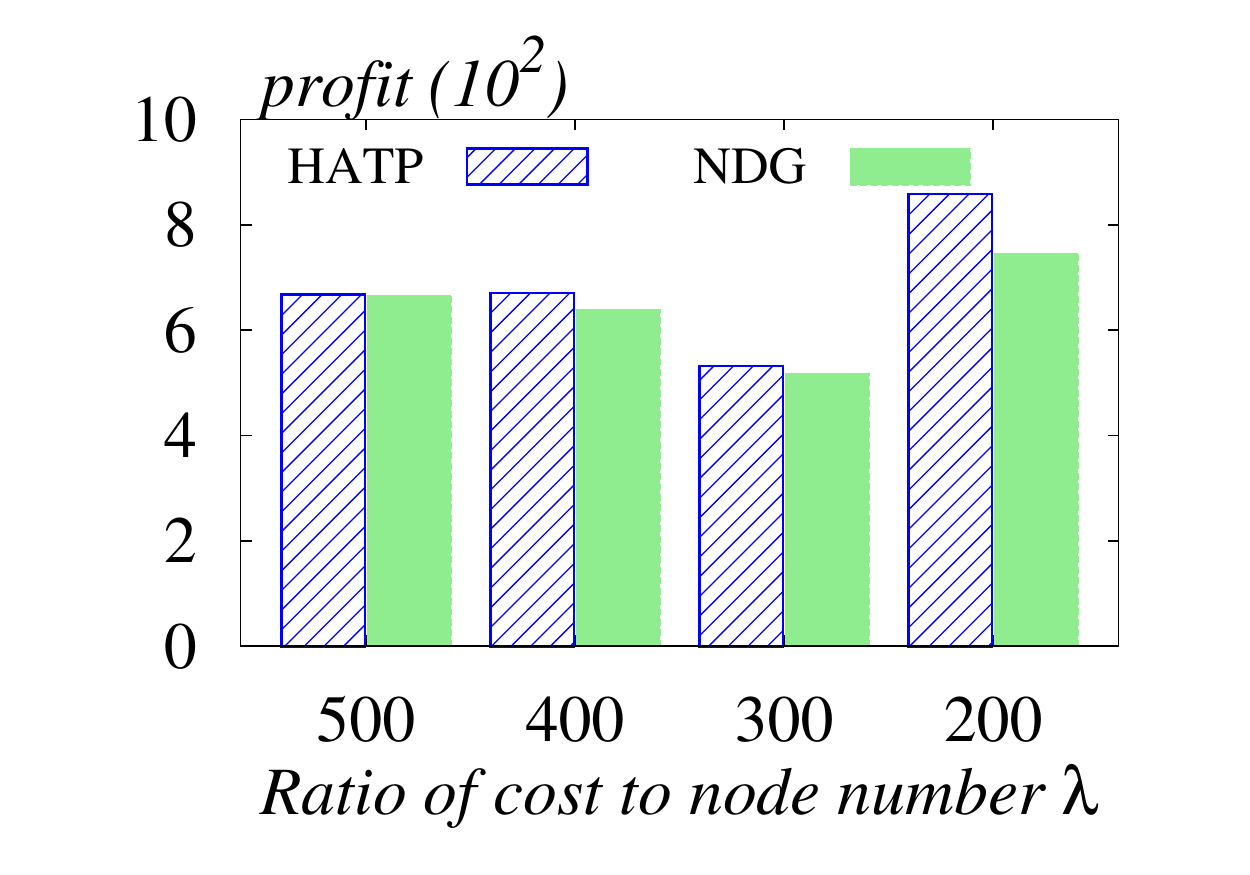}\label{subfig:deg-dg}}\hfill
	\subfloat[Uniform cost]{\includegraphics[width=0.5\linewidth]{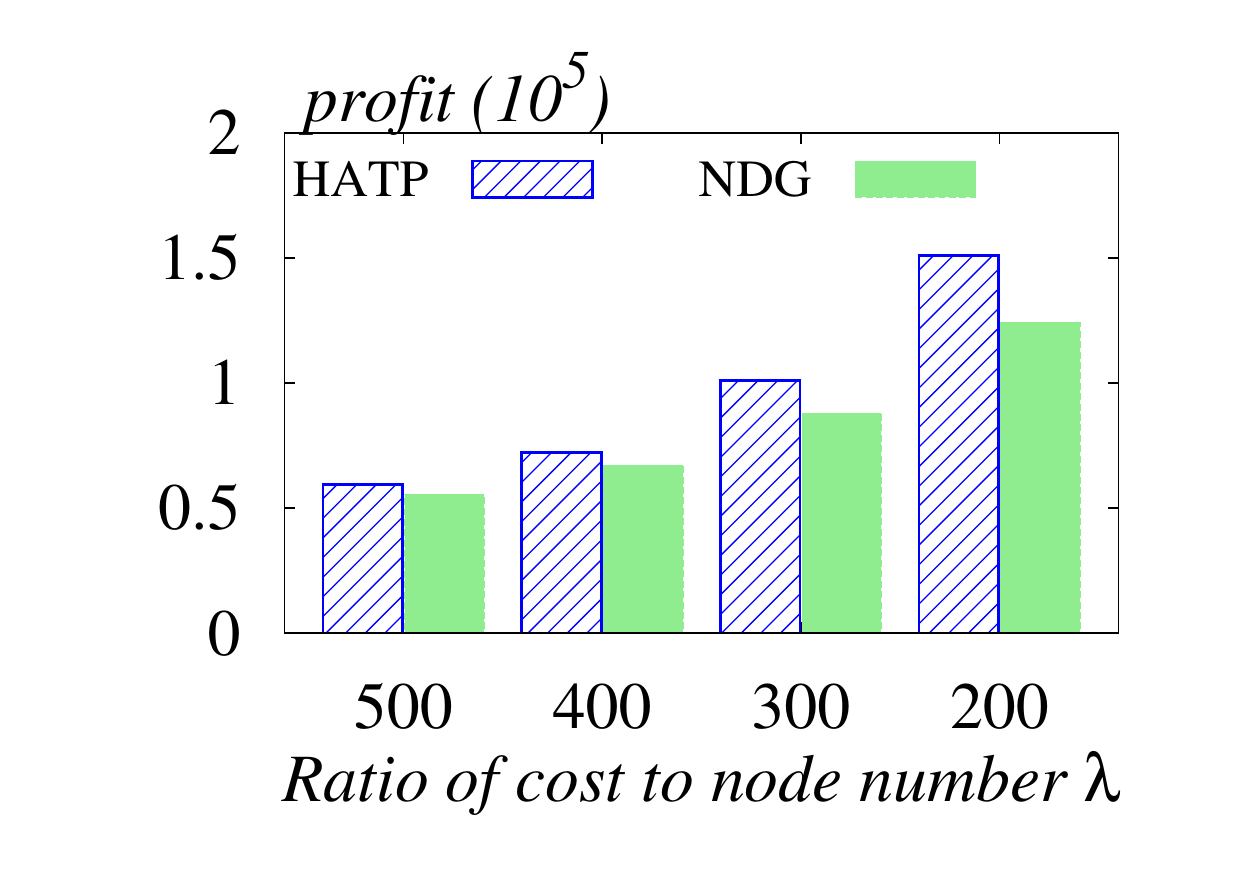}\label{subfig:uni-dg}}
	\caption{Profits of \ATPA and \NDG on {LiveJounral}.}\label{fig:profit-dg}
\end{figure}

\begin{figure}[!t]
	\centering
	\subfloat[Degree-proportional cost]{\includegraphics[width=0.5\linewidth]{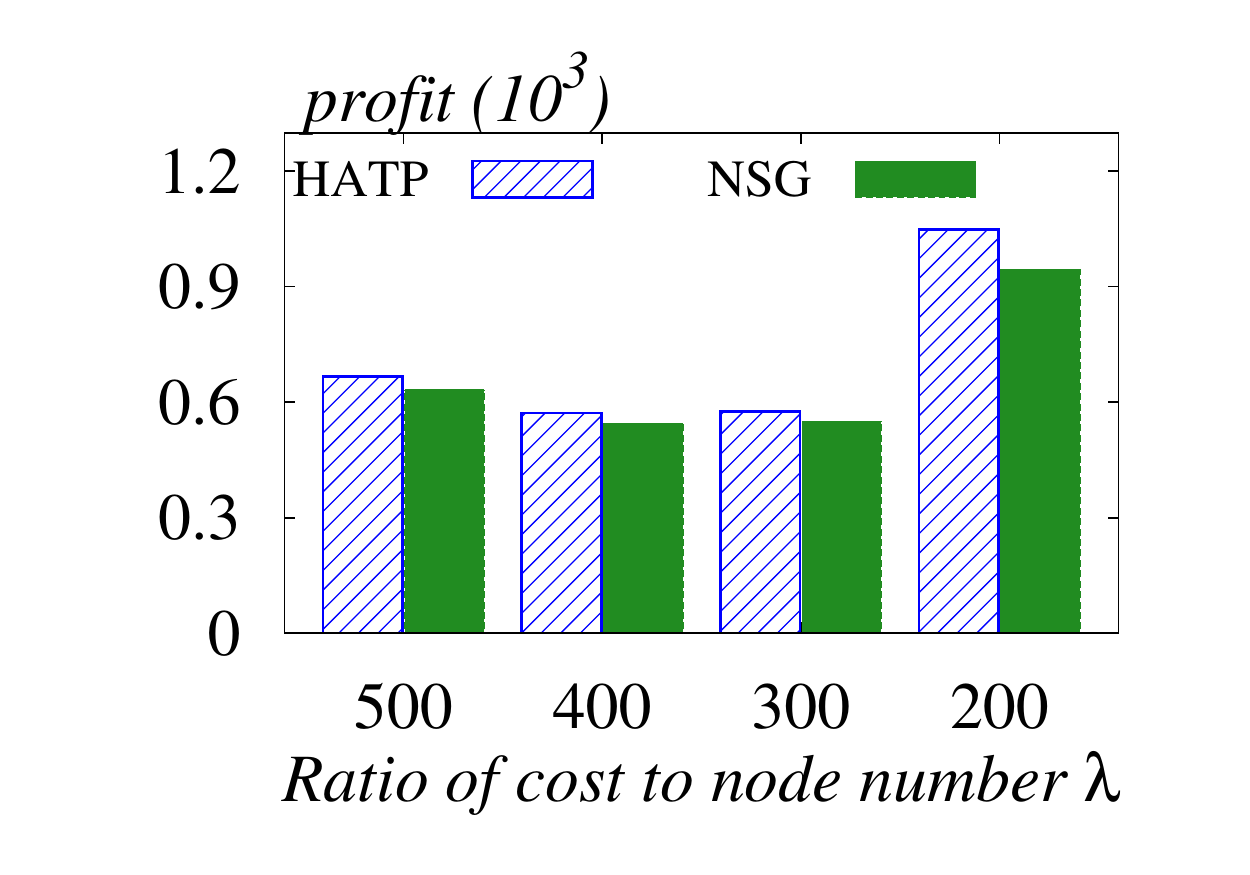}\label{subfig:deg-sg}}\hfill
	\subfloat[Uniform cost]{\includegraphics[width=0.5\linewidth]{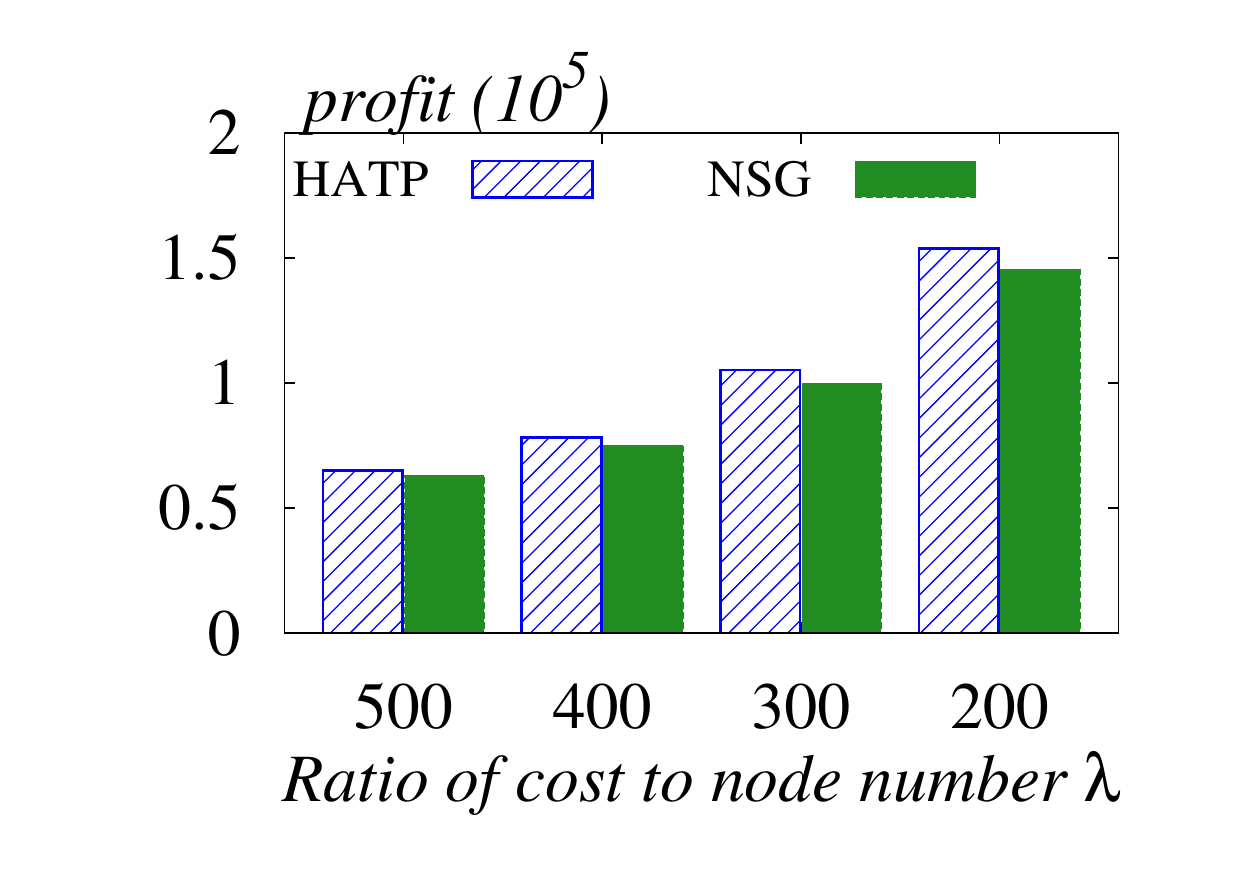}\label{subfig:uni-sg}}
	\caption{Profits of \ATPA and \NSG on {LiveJounral}.}\label{fig:profit-sg}
\end{figure}

\figurename~\ref{fig:profit-dg} presents the profits of \ATPA and \NDG with various $\lambda$ values under the two cost settings. Note that smaller $\lambda$ value means larger seed set size $k$. The overall improvement of \ATPA over \NDG is around $10\%$ and $15\%$ under the degree-proportional and uniform cost settings respectively. In particular, the improvement ratio can be up to $21.3\%$ with $\lambda=200$ in \figurename~\ref{subfig:uni-dg}. The advantage of \ATPA over \NDG in \figurename~\ref{subfig:uni-dg} gets more notable along the decrease of $\lambda$, which reveals that larger the target set size is, more effective the adaptive algorithms become. This can also explain why the advantage in \figurename~\ref{subfig:uni-dg} is more obvious than that in \figurename~\ref{subfig:deg-dg}, since the target set size under the uniform cost setting is larger with the same value of $\lambda$. \figurename~\ref{fig:profit-sg} shows the profit improvement of \ATPA over \NSG. The overall improvement of \ATPA over \NSG is $5\%$ or so, less significant compared with the results in \figurename~\ref{fig:profit-dg}. This indicates that \NSG might be more effective than \NDG on profit maximization. Observe that the improvement is more notable in \figurename~\ref{subfig:uni-sg} than that in \figurename~\ref{subfig:deg-sg}, which again verifies the fact that the advantage of adaptive algorithms over nonadaptive algorithms is more impressive when the target size gets larger.

\begin{figure}[!t]
	\centering
	\subfloat[Running time]{\includegraphics[width=0.5\linewidth]{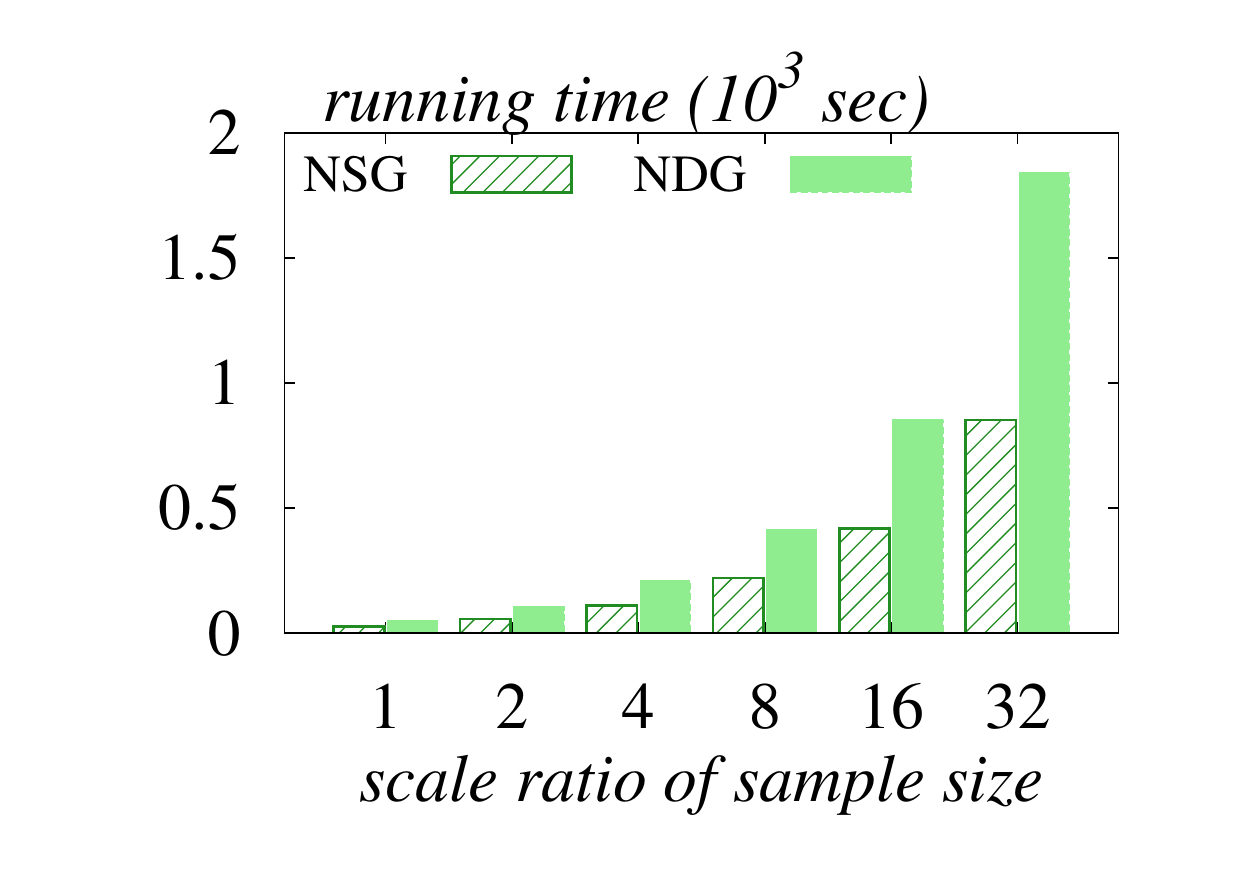}\label{subfig:varytime}}\hfill
	\subfloat[Profit]{\includegraphics[width=0.5\linewidth]{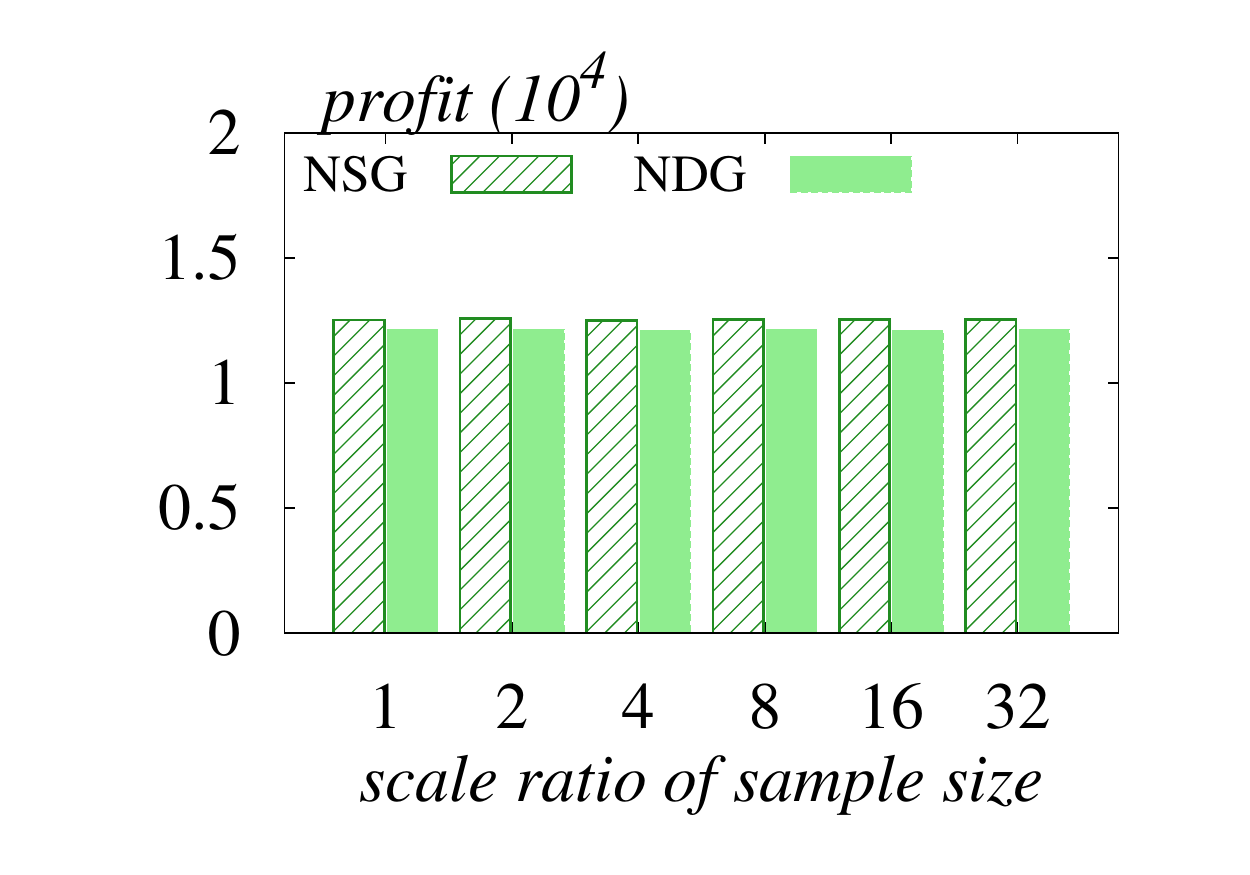}\label{subfig:varyprofit}}
	\caption{\NSG and \NDG with various sample sizes on {Epinions}.}\label{fig:varytimeprofit}
\end{figure}

To further verify the advantage of adaptive algorithms over nonadaptive algorithms, we increase the sample size of \NSG and \NDG with a multiplicative factor of $\{1,2,4,8,16,32\}$ on the {Epinions} dataset with $k=500$ under the degree-proportional cost setting. \figurename~\ref{subfig:varytime} shows that the running time of both \NSG and \NDG increases linearly along with the sample size. However, as shown in \figurename~\ref{subfig:varyprofit}, the profits achieved by \NSG and \NDG almost remain the same (with negligible changes) when we increase the sample size, which is also confirmed in existing work~\cite{Tang_TIM_2014, Huang_SSA_2017, Tang_OPIM_2018} that the spreads cannot be improved after the sample size reaches certain threshold. This indicates that the superiority of our adaptive algorithms over the nonadaptive algorithms is due to the adaptive techniques rather than the sample size.
\section{Conclusion}\label{sec:conclusion}

This paper studies the adaptive target profit maximization (TPM) problem which aims to identify a subset of nodes from the target node set to maximize the expected profit utilizing the adaptive advantage. To acquire an overall understanding on adaptive TPM problem, we investigate this problem in both oracle model and noise model. We prove that adaptive double greedy could address adaptive TPM under the oracle model and achieve an approximation ratio of $\frac{1}{3}$. To address it under the noise model, we design \ATPS that achieves provable approximation guarantee. We later optimize \ATPS into \ATPA that has made remarkable improvement on efficiency. \ATPA is designed based on the concept of hybrid error that could efficiently handle different nodes with various expected marginal spread by adjusting the relative error and additive error adaptively. To evaluate the performance of \ATPA, we conduct extensive experiments on real social networks, and the experimental results strongly confirm the superiorities and effectiveness of our approaches.

\section*{Acknowledgment}
This research is supported by Singapore National Research Foundation under grant~NRF-RSS2016-004, by Singapore Ministry of Education Academic Research Fund Tier 2 under grant~MOE2015-T2-2-069, and by National University of Singapore under an SUG.

\begin{appendix}
	\section{Appendix}\label{sec:appendix}

\subsection{Proofs of Lemmas and Theorems}\label{append:proofs}

\begin{proof}[Proof of Lemma~\ref{lem:case1+2}]
Consider the case that $\rho_\mathrm{f}\ge \rho_\mathrm{r}$. Then, we have $S_i=S_{i-1}\cup \{u_i\}$, $T_i=T_{i-1}$, and $S^\circ_{i}=S^\circ_{i-1}\cup\{u_i\}$. Thus, \eqref{eqn:corestep} becomes $-\Delta_{G_i}(u_i\mid S^\circ_{i-1}) \leq \Delta_{G_i}(u_i\mid S_{i-1}) =\rho_\mathrm{f}$.
%\begin{equation*}
%-\Delta_{G_i}(u_i\mid S^\circ_{i-1}) \leq \Delta_{G_i}(u_i\mid S_{i-1}) =\rho_\mathrm{f}.
%\end{equation*}
If $u_i\notin S^\circ_{i-1}$, we have $-\Delta_{G_i}(u_i\mid S^\circ_{i-1}) \leq -\Delta_{G_i}(u_i\mid T_{i-1} \setminus \{u_i\})=\rho_\mathrm{r}\leq \rho_\mathrm{f}$ due to the submodularity of $\rho_{G_i}(\cdot)$ and the fact that $S^\circ_{i-1}\subseteq T_{i-1} \setminus \{u_i\}$ and $u_i\notin T_{i-1} \setminus \{u_i\}$. If $u_i\in S^\circ_{i-1}$, we have $-\Delta_{G_i}(u_i\mid S^\circ_{i-1})=0\leq \rho_\mathrm{f}$ since $\rho_\mathrm{f}+\rho_\mathrm{r}\geq 0$ by Lemma~\ref{lem:frsum}. Thus, it holds that $-\Delta_{G_i}(u_i\mid S^\circ_{i-1}) \leq\rho_\mathrm{f}$.

The other case of $\rho_\mathrm{f}< \rho_\mathrm{r}$ is analogous. Specifically, we have $S_i=S_{i-1}$, $T_i=T_{i-1}\setminus\{u_i\}$, and $S^\circ_{i}=S^\circ_{i-1}\setminus\{u_i\}$ and need to show that $\rho_{G_i}(S^\circ_{i-1}) - \rho_{G_i}(S^\circ_{i})\leq \rho_\mathrm{r}$. Again, if $u_i\notin S^\circ_{i-1}$, we have $\rho_{G_i}(S^\circ_{i-1}) - \rho_{G_i}(S^\circ_{i})=0\leq \rho_\mathrm{r}$ since $\rho_\mathrm{f}+\rho_\mathrm{r}\geq 0$. If $u_i\in S^\circ_{i-1}$, we have $\rho_{G_i}(S^\circ_{i-1}) - \rho_{G_i}(S^\circ_{i})=\Delta_{G_i}(u_i\mid S^\circ_{i-1}\setminus\{u_i\})\leq  \Delta_{G_i}(u_i\mid S_{i-1})=\rho_\mathrm{f} < \rho_\mathrm{r}$ since $S_{i-1}\subseteq S^\circ_{i-1} \setminus \{u_i\}$ and $u_i\notin S^\circ_{i-1} \setminus \{u_i\}$.

Hence, the lemma is proved.
\end{proof}

\begin{proof}[Proof of Lemma~\ref{lem:policyprofit}]
For simplicity, let $S_i:=S_\phi(\pi^\mathrm{f}_{[i]})$, $T_i:=S_\phi(\pi^\mathrm{r}_{[i]})$, and $S^\circ_{i}:=S_\phi(\pi^\circ_{[i]})$. By definition, we have
\begin{align}
	&\Lambda(\pi^\mathrm{f}_{[i]}) - \Lambda(\pi^\mathrm{f}_{[i-1]})=\sum_{\phi}\!\big(\rho_\phi(S_{i})-\rho_\phi(S_{i-1})\big)\cdot p(\phi),\label{eqn:lambda-f}\\
	&\Lambda(\pi^\mathrm{r}_{[i]}) - \Lambda(\pi^\mathrm{r}_{[i-1]})=\sum_{\phi}\!\big(\rho_\phi(T_{i})-\rho_\phi(T_{i-1})\big)\cdot p(\phi),\label{eqn:lambda-r}\\
	&\Lambda(\pi^\circ_{[i]}) - \Lambda(\pi^\circ_{[i-1]})=\sum_{\phi}\!\big(\rho_\phi(S^\circ_{i})-\rho_\phi(S^\circ_{i-1})\big)\cdot p(\phi).\label{eqn:lambda-o}
\end{align}

Let $\mathcal{G}_i$ be the distribution of $i$-th residual graph with respect to the \ADG policy, where each residual graph $G_i\in\mathcal{G}_i$ has a probability of $\Pr[G_i]$. Then, the realization space $\Omega$ can be disjointedly partitioned with respect to each $G_i\in\mathcal{G}_i$. For each $G_i$, we denote the corresponding set of realizations as $\Omega(G_i)$, which implies that $\Pr[G_i]=\sum_{\phi\in \Omega(G_i)}p(\phi)$. Moreover, for any $S$ and $G_i$, we have 
\begin{equation*}
	\rho_{G_i}(S)=\sum_{\phi\in \Omega(G_i)} \rho_{\phi}(S)\cdot \frac{p(\phi)}{\Pr[G_i]}.
\end{equation*}
Then, \eqref{eqn:lambda-f}--\eqref{eqn:lambda-o} can be rewritten as
\begin{align*}
&\Lambda(\pi^\mathrm{f}_{[i]}) - \Lambda(\pi^\mathrm{f}_{[i-1]})=\sum_{G_i}\!\big(\rho_{G_i}(S_{i})-\rho_{G_i}(S_{i-1})\big)\cdot \Pr[G_i],\\
&\Lambda(\pi^\mathrm{r}_{[i]}) - \Lambda(\pi^\mathrm{r}_{[i-1]})=\sum_{G_i}\!\big(\rho_{G_i}(T_{i})-\rho_{G_i}(T_{i-1})\big)\cdot \Pr[G_i],\\
&\Lambda(\pi^\circ_{[i]}) - \Lambda(\pi^\circ_{[i-1]})=\sum_{G_i}\!\big(\rho_{G_i}(S^\circ_{i})-\rho_{G_i}(S^\circ_{i-1})\big)\cdot \Pr[G_i].
\end{align*}
Putting it all together with Lemma~\ref{lem:case1+2} completes the proof.
\end{proof}

\begin{proof}[Proof of Lemma~\ref{lem:rho-noise}]
Let $f_\mathrm{est}$ and $r_\mathrm{est}$ be the estimations of $\E[I_{G_i}(u_i\mid S_{i-1})]$ and $\E[I_{G_i}(u_i\mid T_{i-1} \setminus \{u_i\})]$. For the $i$-th iteration of \ATPS and the $j$-th round estimations of expected marginal spreads (or profits), we define two events $\mathcal{E}_{i,j}^1$ and $\mathcal{E}_{i,j}^2$ as
\begin{align*}
	&\mathcal{E}_{i,j}^1\colon |\E[I_{G_i}(u_i\mid S_{i-1})]-f_\mathrm{est}|\leq n_i\zeta_i,\\
	&\mathcal{E}_{i,j}^2\colon |\E[I_{G_i}(u_i\mid T_{i-1} \setminus \{u_i\})]-r_\mathrm{est}|\leq n_i\zeta_i.
\end{align*}
Regarding the event $\mathcal{E}_{i,j}^1$ or $\mathcal{E}_{i,j}^2$, we know that $\delta_i=\frac{1}{2^{j-1}kn}$ and $\theta=\frac{\ln(8/\delta_i)}{2\zeta^2_i}$. Thus, by Lemma~\ref{lem:hoeffding}, we have
\begin{equation*}
	\Pr[\neg \mathcal{E}_{i,j}^1\vee \neg \mathcal{E}_{i,j}^2]\leq \Pr[\neg \mathcal{E}_{i,j}^1] +\Pr[\neg \mathcal{E}_{i,j}^2]\leq 4\e^{-2\theta\zeta^2_i}=\frac{1}{2^{j}kn}. 
\end{equation*}
As a result, we have
\begin{equation}\label{eqn:prob}
\Pr\Big[\bigvee\nolimits_{j=1}^{\infty}\big(\neg \mathcal{E}_{i,j}^1\vee \neg \mathcal{E}_{i,j}^2\big)\Big]\leq \sum\nolimits_{j=1}^{\infty} \frac{1}{2^{j}kn}=\frac{1}{kn}.
\end{equation}

In what follows, we assume that events $\mathcal{E}_{i,j}^1$ and $\mathcal{E}_{i,j}^2$ happen for every $j$. Thus, we have $\fest-n_i\zeta_i\leq \rho_\mathrm{f}\leq \fest+n_i\zeta_i$ and $\rest-n_i\zeta_i\leq \rho_\mathrm{r}\leq \rest+n_i\zeta_i$. We first consider the case that $C_1$ occurs. If $\fest\geq \rest$, we have $\fest-\rest\geq 2n_i\zeta_i$ or $\rest\leq -n_i\zeta_i$ since $\fest+\rest\geq 0$ (using the same argument in the proof of Lemma~\ref{lem:frsum}). For the former, $\rho_\mathrm{f} -\rho_\mathrm{r}\geq \fest-\rest-2n_i\zeta_i\geq 0$, while for the latter, $\rho_\mathrm{r}\leq 0$ and hence $\rho_\mathrm{f}\geq \rho_\mathrm{r}$ as $\rho_\mathrm{f}+\rho_\mathrm{r}\geq 0$ by Lemma~\ref{lem:frsum}. Thus, $\rho_\mathrm{f} \geq \rho_\mathrm{r}$ always hold if $\fest\geq \rest$. With a similar analysis, $\rho_\mathrm{f} <\rho_\mathrm{r}$ also holds if $\fest< \rest$. Then, when $C_1$ occurs, by Lemma~\ref{lem:case1+2}, we have
\begin{equation*}
\begin{split}
&\rho_{G_i}(S^\circ_{i-1}) - \rho_{G_i}({S}^\circ_{i})\\
&\leq \rho_{G_i}({S}_{i})-\rho_{G_i}(S_{i-1}) + \rho_{G_i}({T}_{i}) - \rho_{G_i}(T_{i-1}). 
\end{split}
\end{equation*} 
Now, we consider the case that $C_2$ occurs. Again, if $\fest\geq \rest$, we have $\rho_\mathrm{r}\leq \rest+n_i\zeta_i \leq \fest+n_i\zeta_i\leq \rho_\mathrm{f} + 2n_i\zeta_i \leq  \rho_\mathrm{f}+2$. On the other hand, one can verify that $\rho_{G_i}(S^\circ_{i-1}) - \rho_{G_i}(S^\circ_{i}) \leq \max\{\rho_\mathrm{f}, \rho_\mathrm{r}\}\leq \rho_\mathrm{f} + 2= \rho_{G_i}(S_{i})-\rho_{G_i}(S_{i-1}) + 2$ and $\rho_{G_i}(T_{i}) - \rho_{G_i}(T_{i-1})=0$. Similarly, if $\fest < \rest$, we have $\rho_{G_i}(S^\circ_{i-1})-\rho_{G_i}(S^\circ_{i}) \leq \max\{\rho_\mathrm{f}, \rho_\mathrm{r}\}\leq \rho_\mathrm{r} + 2 \leq \rho_{G_i}(T_{i}) - \rho_{G_i}(T_{i-1}) + 2$ and $\rho_{G_i}(S_{i})-\rho_{G_i}(S_{i-1})=0$. Then,
\begin{equation}\label{eqn:true}
\begin{split}
&\rho_{G_i}(S^\circ_{i-1}) - \rho_{G_i}({S}^\circ_{i})-2\\
&\leq \rho_{G_i}({S}_{i})-\rho_{G_i}(S_{i-1}) + \rho_{G_i}({T}_{i}) - \rho_{G_i}(T_{i-1}),
\end{split}
\end{equation} 
which also holds for the case that $C_1$ occurs.

On the other hand, suppose that at least one of the events events $\mathcal{E}_{i,j}^1$ and $\mathcal{E}_{i,j}^2$ does not happen. We know that $\rho_{G_i}(S^\circ_{i-1})-\rho_{G_i}(S^\circ_{i}) \leq \max\{\rho_\mathrm{f}, \rho_\mathrm{r}\}\leq n$ and $\rho_{G_i}({S}_{i})-\rho_{G_i}(S_{i-1}) + \rho_{G_i}({T}_{i}) - \rho_{G_i}(T_{i-1})\geq \min\{\rho_\mathrm{f}, \rho_\mathrm{r}\}\geq -n$. Thus, it always holds that
\begin{equation}\label{eqn:false}
\begin{split}
&\rho_{G_i}(S^\circ_{i-1}) - \rho_{G_i}({S}^\circ_{i})-2n\\
&\leq \rho_{G_i}({S}_{i})-\rho_{G_i}(S_{i-1}) + \rho_{G_i}({T}_{i}) - \rho_{G_i}(T_{i-1}),
\end{split}
\end{equation}

Combining \eqref{eqn:prob}, \eqref{eqn:true}, and \eqref{eqn:false} completes the proof.
\end{proof}

\begin{proof}[Proof of Lemma~\ref{lem:policyprofitadd}]
We first fix a random seed for \ATPS such that \ATPS becomes a deterministic algorithm. Using a similar argument in the proof of Lemma~\ref{lem:policyprofit}, we have
\begin{align*}
&\Lambda(\pi^\mathrm{f}_{[i]}) - \Lambda(\pi^\mathrm{f}_{[i-1]})=\sum_{G_i}\!\big(\rho_{G_i}(S_{i})-\rho_{G_i}(S_{i-1})\big)\cdot \Pr[G_i],\\
&\Lambda(\pi^\mathrm{r}_{[i]}) - \Lambda(\pi^\mathrm{r}_{[i-1]})=\sum_{G_i}\!\big(\rho_{G_i}(T_{i})-\rho_{G_i}(T_{i-1})\big)\cdot \Pr[G_i],\\
&\Lambda(\pi^\circ_{[i]}) - \Lambda(\pi^\circ_{[i-1]})=\sum_{G_i}\!\big(\rho_{G_i}(S^\circ_{i})-\rho_{G_i}(S^\circ_{i-1})\big)\cdot \Pr[G_i].
\end{align*}
Then, taking the expectation on both sides over the internal randomness of the \ATPS algorithm and combing with Lemma~\ref{lem:rho-noise} complete the proof.
\end{proof}
\begin{proof}[Proof of Lemma~\ref{lem:atpacase3}]
The proof is analogous to those of Lemmas~\ref{lem:rho-noise} and \ref{lem:policyprofitadd}. The key step is to figure out the bound on profit loss when $C_2^\prime$ occurs, \ie~$n_i\zeta_i\leq 1$ and $\varepsilon_i\leq \varepsilon$, and the estimation is sufficiently accurate (with probability $1-1/(kn)$), \ie~$\frac{f_\mathrm{est} -n_i\zeta_i}{1+\varepsilon_i}-c(u_i) \le \rho_\mathrm{f} \le \frac{f_\mathrm{est} +n_i\zeta_i}{1-\varepsilon_i}-c(u_i)$ and $c(u_i)-\frac{r_\mathrm{est} +n_i\zeta_i}{1-\varepsilon_i}\le \rho_\mathrm{r} \le c(u_i)-\frac{r_\mathrm{est} -n_i\zeta_i}{1+\varepsilon_i}$. If $f_\mathrm{est}+r_\mathrm{est}\ge 2c(u_i)$, one can verify that $ \rho_\mathrm{r}\leq \frac{f_\mathrm{est} -n_i\zeta_i}{1+\varepsilon}-c(u_i) + \frac{2(n_i\zeta_i+\varepsilon c(u_i))}{1-\varepsilon}\leq \rho_\mathrm{f}+\frac{2(1+\varepsilon c(u_i))}{1-\varepsilon}$. Therefore, we have  $\rho_{G_i}(S^\circ_{i-1}) - \rho_{G_i}(S^\circ_{i}) \leq \max\{\rho_\mathrm{f}, \rho_\mathrm{r}\}\leq \rho_\mathrm{f} + \frac{2(1+\varepsilon c(u_i))}{1-\varepsilon} = \rho_{G_i}(S_{i})-\rho_{G_i}(S_{i-1}) + \frac{2(1+\varepsilon c(u_i))}{1-\varepsilon}$ and $\rho_{G_i}(T_{i}) - \rho_{G_i}(T_{i-1})=0$. Similarly, if $f_\mathrm{est}+r_\mathrm{est} < 2c(u_i)$, we have $\rho_{G_i}(S^\circ_{i-1})-\rho_{G_i}(S^\circ_{i}) \leq \max\{\rho_\mathrm{f}, \rho_\mathrm{r}\}\leq \rho_\mathrm{r} + \frac{2(1+\varepsilon c(u_i))}{1-\varepsilon} \leq \rho_{G_i}(T_{i}) - \rho_{G_i}(T_{i-1}) + \frac{2(1+\varepsilon c(u_i))}{1-\varepsilon}$ and $\rho_{G_i}(S_{i})-\rho_{G_i}(S_{i-1})=0$. Following the same arguments in the proofs of Lemmas~\ref{lem:rho-noise} and \ref{lem:policyprofitadd}, the lemma is proved. 
\end{proof}

\end{appendix}

\balance

\bibliographystyle{IEEEtranS}
\bibliography{reference}

\end{sloppy}

\end{document}